\documentclass[11pt]{article}
\usepackage[margin=1in]{geometry}
\usepackage[utf8]{inputenc}
\usepackage[T1]{fontenc}
\usepackage{authblk}
\usepackage{backref}
\usepackage{todonotes}
\usepackage{xspace}
\usepackage[numbers,sort]{natbib}
\usepackage{enumitem}
\usepackage[normalem]{ulem}

\usepackage{comment}

\usepackage{bm}
\usepackage{nicefrac}
\usepackage{algorithm}
\usepackage{algorithmicx}
\usepackage[noend]{algpseudocode}

\usepackage{amssymb}
\usepackage{amsmath}
\usepackage{amsthm}
\usepackage{dsfont}
\usepackage{xcolor}
\usepackage{nameref}
\usepackage{url}
\definecolor{ForestGreen}{rgb}{0.1333,0.5451,0.1333}
\definecolor{DarkRed}{rgb}{0.65,0,0}
\definecolor{Red}{rgb}{1,0,0}
\usepackage[linktocpage=true,backref=page,
pagebackref=true,colorlinks,
linkcolor=DarkRed,citecolor=ForestGreen,
bookmarks,bookmarksopen,bookmarksnumbered]
{hyperref}

\usepackage[all]{hypcap}
\usepackage{footnote}
\usepackage{thm-restate}
\usepackage[tableposition=bottom]{caption}
\usepackage[all]{hypcap}
\usepackage{subcaption}
\usepackage{framed}
\usepackage[framemethod=tikz]{mdframed}
\usepackage[skins]{tcolorbox}

\makeatletter
\g@addto@macro{\maketitle}{\@thanks}
\makeatother

\usepackage{cleveref}
\theoremstyle{plain}
\newtheorem{thm}{Theorem}[section]
\newtheorem{cor}[thm]{Corollary}

\newtheorem{fact}[thm]{Fact}
\newtheorem{lem}[thm]{Lemma}
\newtheorem{Def}[thm]{Definition}
\newtheorem{obs}[thm]{Observation}
\newtheorem{claim}[thm]{Claim}

\newtheorem{remark}[thm]{Remark}

\newcommand{\E}{\mathbb{E}}%
\newcommand{\Ber}{\mathrm{Ber}}
\newcommand{\Var}{\mathrm{Var}}
\newcommand{\Cov}{\mathrm{Cov}}

\newcommand{\eps}{\epsilon}%
\newcommand{\poly}{\mathrm{poly}}

\renewcommand{\algorithmiccomment}[1]{\bgroup\hfill$\rhd$~#1\egroup}

\newcounter{note}[section]

\newcommand{\calA}{\mathcal{A}}
\newcommand{\calB}{\mathcal{B}}
\newcommand{\calD}{\mathcal{D}}
\newcommand{\calI}{\mathcal{I}}
\newcommand{\calH}{\mathcal{H}}
\newcommand{\calM}{\mathcal{M}}

\newcommand{\calS}{\mathcal{S}}

\renewcommand{\vec}[1]{\bf{#1}}

\newenvironment{wrapper}[1]
{
	\begin{center}
		\begin{minipage}{\linewidth}
			\begin{mdframed}[hidealllines=true, backgroundcolor=gray!20, leftmargin=0cm,innerleftmargin=0.4cm,innerrightmargin=0.4cm,innertopmargin=0.4cm,innerbottommargin=0.4cm,roundcorner=10pt]
				#1}
			{\end{mdframed}
		\end{minipage}
	\end{center}
} 

\renewcommand{\paragraph}[1]{\medskip\noindent\textbf{#1}}

\newcommand{\ratio}{0.652}
\newcommand{\ratiobmatching}{0.646}

\title{Online Dependent Rounding Schemes for Bipartite Matchings, with Applications}
\author[1]{Joseph (Seffi) Naor\thanks{Supported in part by ISF grant 3001/24 and United States - Israel BSF grant 2022418.}}
\author[2]{Aravind Srinivasan\thanks{Supported in part by NSF award number CCF-1918749, and by research awards from Amazon and Google.}}
\author[1]{David Wajc\footnote{Supported in part by a Taub Family Foundation “Leader in Science and Technology” fellowship, and by ISF grant 3200/24, ``Prophets, Philosophers, and Online Algorithms''. Work done in part while at Stanford University.}}
\affil[1]{Technion -- Israel Institute of Technology}
\affil[2]{University of Maryland, College Park}
\date{}

\begin{document}

	\maketitle

	\pagenumbering{gobble}
	\begin{abstract}
We introduce the abstract problem of rounding an unknown fractional bipartite $b$-matching $\bf{x}$ revealed online (e.g., output by an online fractional algorithm), exposed node-by-node on~one~side.
The objective is to maximize the \emph{rounding ratio} of the output matching $\mathcal{M}$, which is the minimum over all fractional $b$-matchings $\bf{x}$, and edges $e$, of the ratio $\Pr[e\in \mathcal{M}]/x_e$.
In analogy with the highly influential offline dependent rounding schemes of Gandhi et al.~(FOCS'02, J.ACM'06), we refer to such algorithms as \emph{online dependent rounding schemes} (ODRSes). 
This problem, with additional restrictions on the possible inputs $\bf{x}$, has played a key role in recent developments in online computing.

We provide the first generic $b$-matching ODRSes that impose no restrictions on $\bf{x}$. 
Specifically, we provide ODRSes with rounding ratios of $\ratiobmatching$ and $\ratio$ for $b$-matchings and simple matchings, respectively. 
This breaks the natural barrier of $1-1/e$, prevalent for online matching problems, and numerous online problems more broadly. Using our ODRSes, we provide a number of algorithms with similar better-than-$(1-1/e)$ ratios for several problems in online edge coloring, stochastic optimization, and more.

Our techniques, which have already found applications in several follow-up works (Patel and Wajc SODA'24, Blikstad et al.~SODA'25, Braverman et al.~SODA'25, and Aouad et al.~2024), include periodic use of \emph{offline} contention resolution schemes (in online algorithm design), grouping nodes, and a new scaling method which we call \emph{group discount and individual markup}.
\end{abstract}
	
	\newpage 
	\tableofcontents
	\newpage
	
	\pagenumbering{arabic}
	\section{Introduction}
We initiate the study of the abstract problem of online \emph{randomized rounding of fractional $b$-matchings}. Here, nodes on one side of an unknown bipartite graph $G=(V,E)$ are revealed over time; each such \emph{online node} and its edges' associated fractions in a fractional $b$-matching $\bf{x}$ are revealed online (e.g., provided by some online fractional algorithm). 
An online dependent rounding scheme (ODRS) must decide immediately and irrevocably, upon a node's arrival, which of its edges to match.\footnote{As in the offline dependent rounding schemes of \cite{gandhi2006dependent}, the term \emph{dependent} here underscores the need for random choices to depend on each other, since independent coin tosses do not yield a valid high-value ($b$-)matching.}
The objective is to match each with probability close to the value associated with it in $\bf{x}$.

\begin{Def}
    An ODRS $\calA$ for online matching has \emph{(oblivious) rounding ratio} $\alpha\in [0,1]$ if, when a vertex with edges $e_1,\dots,e_n$ are revealed online with their associated fractions $x_{e_1},\dots,x_{e_n}$, with the guarantee that ${\bf{x}}\in \{{\bf{x}}\in \mathbb{R}^{|E|}_{\geq 0} \mid x(\delta(v))\leq 1 \; \forall  
 v\in V\}$ is a fractional matching, $\calA$ outputs an \emph{integral} matching $\calM$ satisfying 
    $$\Pr[e \in \calM]\geq \alpha\cdot x_{e}\quad  \forall e\in E.$$
\end{Def}

In offline settings a rounding ratio of one is achievable, by the integrality of the bipartite matching polytope. In contrast, in online settings, this is provably impossible \cite{devanur2013randomized,cohen2018randomized,buchbinder2023lossless}.
Nonetheless, as illustrated by the FOCS 2023 workshop on online algorithms and online rounding,\footnote{\href{https://sites.google.com/view/focs23workshop-online-rounding/home}{Online Algorithms
and Online Rounding: Recent Progress} @ FOCS 2023.} 
algorithms for rounding fractional matchings online are central to numerous breakthroughs in online computing in recent years, for problems such as online bipartite edge-weighted matching \cite{fahrbach2020edge,blanc2021multiway,gao2021improved}, stochastic online bipartite matching \cite{tang2022fractional,huang2022power}, and online edge coloring \cite{cohen2019tight,kulkarni2022online,blikstad2024online}, among others.
Despite the impossibility of rounding ratios of one or $1-o(1)$ in general, many of the above results follow from ODRS with such high rounding ratios for 
\emph{structured} fractional matchings.
Echoing these recent developments, Buchbinder et al.~\cite{buchbinder2023lossless} asked what additional constraints can be imposed on fractional matchings to allow for lossless (i.e., rounding ratio of one) online rounding.

Breaking from prior work, we explicitly study the problem of designing ODRSes for \emph{unrestricted} fractional ($b$-)matchings, and their applications.
While recent work of \cite{saberi2021greedy} implies an ODRS with non-trivial rounding ratio of $0.526$ for matchings (but not for $b$-matchings), our understanding of such rounding schemes is nascent, both on the techniques side, as well as the applications side.

\paragraph{Our ODRSes.} We provide a number of ODRSes, for both structured and unrestricted fractional $b$-matchings, as well as a host of applications. 

For star graphs, or equivalently rounding in a uniform matroid (i.e., fractions summing up to at most some $b$ are revealed online), we provide two simple ODRSes with rounding ratio of one respecting prefix constraints, thus providing a counterpart to the classic pivotal sampling of \cite{srinivasan2001distributions}. Indeed, as we show (in \Cref{alg:online-rounding-properties}), one of our algorithm's output distribution is \emph{identical} to the offline algorithm's output distribution, and thus inherits the latter's known strong negative correlation properties, which we leverage later.

For ODRS for unrestricted fractional matchings, we rule out any better than $2\sqrt{2}-2\approx 0.828$ rounding ratio.
On the positive side, we show that combining ODRS for star gaphs with the single-item \emph{offline} contention resolution scheme (CRS) of Feige and Vondr\'ak \cite{feige2006allocation}
yields a simple ODRS with rounding ratio of $1-1/e\approx 0.632$. Our main result breaks this ubiquitous bound.

\begin{wrapper}
\begin{thm}\label{thm:arbitrary}
    There exist bipartite matching ($b$-matching) ODRSes with rounding ratio $\ratio$ ($\ratiobmatching$).
    No bipartite matching ODRS has rounding ratio greater than $2\sqrt{2}-2\approx 0.828$.
\end{thm}
\end{wrapper}

\subsection{Extensions and Applications}\label{sec:applications}
Our ODRSes and the techniques underlying them have a number of applications. 
We briefly discuss our three main applications here, deferring detailed discussions to \Cref{sec:applications-body}, where we also present other applications of our ODRS to fairness in machine learning and to algorithms with predictions.

\paragraph{Online edge coloring.} 
Our ability to round fractional matchings has implications to the recent active line of work on online edge coloring \cite{cohen2019tight,saberi2021greedy,kulkarni2022online,bhattacharya2021online,aggarwal2003switch,bahmani2012online,blikstad2024online,blikstad2024simple,dudeja2025randomized,blikstad2025deterministic}.
Here, a graph is revealed online (either vertex-by-vertex or edge-by-edge), and we must decompose this graph into color classes that form matchings (i.e., no vertex has two edges of the same color). 
A reduction due to \cite{cohen2019tight} shows that online $(\alpha+o(1))\Delta$-edge-coloring can be reduced to online matching algorithms that match each edge with probability $1/(\alpha\Delta)$. Ignoring the $o(1)$ term, these problems are in fact equivalent, since sampling a color at random among the $\alpha\Delta$ color yields an online matching algorithm matching each edge with probability $1/(\alpha\Delta)$.
This reduction has motivated the interest in online rounding a particular structured fractional matching: the uniform solution, $\bf{x}\; $$ = \frac{1}{\Delta}\cdot \bf{1}$ \cite{kulkarni2022online,blikstad2024online,cohen2019tight,cohen2018randomized}, culminating in a competitive ratio (and rounding ratio for such ``spread out'' fractional matchings) of $(1+o(1))$ for general (simple) graphs under edge arrivals \cite{blikstad2024online}.
For \emph{multigraphs}, in contrast, positive results were only known under \emph{random-order arrivals} \cite{aggarwal2003switch}.
Addressing this gap between simple graphs and multigraphs, in \Cref{sec:app-edge-coloring} we extend the above reduction from $\alpha$-competitive edge coloring to computing online matchings that match each (in this case, parallel) edge with probability $\frac{1}{\alpha \Delta}$. We thus show that a rounding ratio of $1/\alpha$
implies a competitive ratio of $\alpha+o(1)$ for edge coloring multigraphs (and is, in fact, equivalent to it). Our positive and negative results for ODRS yield the first positive and negative results for online edge-coloring in multigraphs under adversarial arrivals. For example, we achieve a competitive factor of $e/(e-1)-\Omega(1)$ in bipartite multigraphs, answering a question of Schrijver \cite[acknowledgements]{cohen2019tight}.

\paragraph{Stochastic optimization.}
Another application of our techniques concerns the online stochastic bipartite weighted matching problem, much studied in the online Prophet-Inequality literature \cite{ezra2020online,feldman2015combinatorial,dutting2020prophet,krengel1978semiamarts}. (See \Cref{sec:stochastic} for formal definition.) 
Recently, Papadimitriou et al.~\cite{papadimitriou2021online} initiated the study of efficiently approximating the optimum \emph{online} algorithm for this problem. They showed that it is \textsc{pspace}-hard to approximate the optimum algorithm within some universal constant $\beta<1$, and provided a polynomial-time online algorithm that $0.51$-approximates this online optimum. (In contrast, the best competitive ratio, or approximation of the \emph{offline} optimum, is $0.5$ for this problem \cite{ezra2020online}.)
The bound of \cite{papadimitriou2021online} can be improved to $0.526$ using ideas in \cite{saberi2021greedy} and was significantly improved to $1-1/e\approx 0.632$ by \cite{braverman2022max}. All three results are achieved via online rounding. Extending our online rounding approach of \Cref{thm:arbitrary} to such stochastic settings, we improve on the recent bound of \cite{braverman2022max} and break the ubiquitous bound of $1-1/e$, providing an online polytime $\ratio$-approximation of the optimal  (computationally-unbounded) online algorithm.

\paragraph{Multi-Stage stochastic optimization.}
Our ODRS for star graphs and its strong negative correlations allows us to generalize a result of \cite{DBLP:conf/soda/Srinivasan07} for multi-stage stochastic optimization. 
Here, uncertain parameters are stochastic but their distributions get refined over time: actions in earlier stages are cheaper but perhaps less accurate due to the stochasticity, while actions in later stages are more expensive but more accurate. 
(See the survey \cite{swamy2006approximation}.) 
A basic example is the (hyper-)graph covering problem, where sets (hyperedges) are revealed in $k$ stages, and must be covered by bought elements (nodes).
Swamy and Shmoys \cite{DBLP:journals/siamcomp/SwamyS12} presented a $2k$-approximation for this problem, later improved to $2$, in \cite{DBLP:conf/soda/Srinivasan07}, matching the offline hardness of the non-stochastic single-stage problem.
Building on the framework of \cite{DBLP:journals/siamcomp/SwamyS12,DBLP:conf/soda/Srinivasan07}, in \Cref{sec:app-stoch-opt} we generalize this result to allow for sets with \emph{muti-}coverage demands. Our algorithms have approximation ratio of $2$ (and tending to \emph{one} in some settings), and these ratios hold w.h.p.

\subsection{Techniques}
In this section we highlight some of the key ideas that allow us to achieve our main results, namely Theorems \ref{thm:arbitrary} and \ref{alg:online-rounding-properties}, deferring detailed discussions of applications to later sections.

For level-set rounding, we provide two simple ODRS that achieves a rounding ratio of one.
Surprisingly, we show via a coupling argument 
that for this problem, our main online algorithm's distribution is \emph{identical} to that of (offline) pivotal sampling, a.k.a.~Srinivasan sampling \cite{srinivasan2001distributions}. 
This allows us to inherit this well-studied algorithm's concentration properties \cite{srinivasan2001distributions,branden2012negative,dubhashi2007positive}. 
In particular, this implies that our online algorithm's output distribution satisfies the \emph{Strong Rayleigh} property (see \Cref{sec:prelims}), which in turn implies a slew of negative correlation properties, useful for our applications in later sections.

For our impossibility result, we rely on a Ramsey-theoretic probabilistic lemma, whereby in any sufficiently large set of binary variables, some (near-)positive correlation is inevitable. 
An extension of this lemma (\Cref{almost-positive-cylinder-dependence}), of possible independent interest, rules out algorithmic approaches based on guaranteeing strict negative correlation between offline nodes' matched statuses.

\paragraph{ODRSes for $b$-matching.}
Our first ODRS for $b$-matchings more broadly is obtained by interleaving level-set ODRSes with repeated invocations of the \emph{offline} 
 single-item contention resolution scheme (CRS) of Feige and Vondr\'{a}k \cite{feige2006allocation}. (This is, to the best of our knowledge, the first online algorithm to use offline CRSes in such a black-box manner.) Specifically, we run independent copies of our level-set algorithm, one per offline node, thus having each offline node $i$ ``bid'' for arriving online node $t$ independently of other offline nodes with probability $x_{i,t}$, while preserving the $b$-matching constraints of the offline nodes. 
To preserve online nodes' matching constraints, we then use the single-item offline CRS of \cite{feige2006allocation} at each time $t$: this CRS guarantees for such product distributions that $t$ is allocated to at most one buyer, with every buyer $i$ being allocated item $t$ (i.e., $(i,t)$ being matched) with probability at least $(1-1/e)\cdot x_{i,t}$, thus yielding a rounding ratio of $1-1/e$.

To improve on the above, we first note that the $(1-1/e)\cdot x_{it}$ bound is loose if any of the $x_{it}$ fractions are bounded away from $0$. Indeed, this $1-1/e$ factor is of the form 
$$\bigg(1-\prod_{i} (1-x_{it})\bigg)\bigg/\sum_i x_{i,t}\geq 1-1/e,$$
where the inequality is loose if any $x_{it}$ is large.
To beat this $1-1/e$ bound, we therefore simulate the existence of a large $x_{it}$, by \emph{grouping} offline nodes via a bin-packing algorithm, and letting at most one node per group (bin) $B$ bid for $t$.
The obtained distribution over biding offline nodes is no longer a product distribution, but it is strongly negatively correlated, and even negatively associated (see \Cref{sec:prelims}). This allows us to achieve the same $1-1/e$ bound  using CRSes for negatively correlated distributions of Bansal and Cohen \cite{bansal2021contention}.
Moreover, this grouping allows us to beat the $1-1/e$ ratio in some cases: 
In a very concrete sense, it results in effectively a single aggregated offline node bidding with large probability $\sum_{i\in B} x_{i,t}$, thus allowing us to beat the bound of $1-1/e$ if much grouping occurs. Furthermore, thinking of bids as ``costing'' the offline node future matching opportunities, we can even beat the bound of $1-1/e$ after giving nodes in a group a ``group discount'', and having them bid with lower probability than $x_{i,t}$.
This discounting then leaves offline nodes with more budget (probability of not being matched), thus allowing us to charge these nodes more at later times: in particular, when these offline nodes are not grouped, we charge them an ``individual markup'', by having them bid with higher probability than $x_{i,t}$. This then increases the matching probability of those $t$ for which little grouping occurs to beyond $1-\frac{1}{e}$, and the improved rounding ratio follows.

\subsection{Further related work}\label{sec:related}

Online bipartite $b$-matching is a prototypical online problem, tracing its origin to the seminal work of Karp et al.~\cite{karp1990optimal}. Here, nodes on one side of a bipartite graph are revealed over time, and each node $v$ can be matched at most $b_v$ times, with arriving nodes matched immediately and irrevocably. This influential problem has been widely studied in various settings (e.g., fractional/integral, simple/capacitated,  unweighted/vertex-weighted/edge-weighted)  \cite{mehta2007adwords,kalyanasundaram2000optimal,karp1990optimal,aggarwal2011online,fahrbach2020edge,blanc2021multiway,devanur2013randomized,birnbaum2008line,eden2018economic,feldman2009online2} and has been a cornerstone of the online algorithms literature. (See the surveys \cite{huang2024online,devanur2022online,mehta2013online}.) 

A natural approach to round a fractional matching $\bf{x}$ is to \emph{activate} each edge $e$ independently w.p.~$x_e$ and then pick a feasible subset $\calI$ of active edges, such that $\Pr[e\in \calI]\geq \alpha \cdot x_e$. This is obtained by \emph{contention resolution schemes} (CRSes), first studied in offline settings for single-item problems (rank-one matroids) by Feige and Vondr\'{a}k \cite{feige2006allocation} (see \Cref{sec:prelims}), and later for arbitrary matroids by Chekuri et al.~\cite{chekuri2010dependent}.This was extended to online settings by Feldman et al.~\cite{feldman2016online}, and has since become a bedrock of stochastic online algorithms, for prophet inequalities  \cite{feldman2016online}, secretary problems \cite{dughmi2020outer,dughmi2022matroid}, posted-price mechanisms \cite{chawla2010multi,kleinberg2012matroid}, sequential pricing \cite{pollner2022improved}, algorithmic contract design \cite{bechtel2022delegated} and more. 
CRSes and OCRSes have been intensely studied for numerous constraints, inspired by \cite{kleinberg2012matroid}; most related to our work is the active line of work on (O)CRSes for matchings \cite{feldman2016online,guruganesh2018understanding,fu2021random,brubach2021improved}, where the optimal balance ratio (the counterpart of rounding ratio for ODRSes) is still not fully understood, despite recent exciting progress \cite{nuti2023towards,bruggmann2022optimal,pollner2022improved,macrury2023random,macrury2024random}. While this is an exciting and active area, such OCRS-based ODRSes cannot have rounding ratio better better than $0.5$, even for rank-one matroids, by connections to the prophet-inequality problem \cite{ezra2020online,krengel1978semiamarts}, or even $1/e\approx 0.362$, as we assume $\bf{x}$ is unknown upfront, and would therefore require \emph{oblivious} OCRS \cite{fu2022oblivious}.
Going beyond independent bids for each pair $(i,t)$ is therefore crucial in order to obtain matching ODRSes with high oblivious rounding ratios.

Another online matching rounding scheme is provided by the recent exciting literature on Online Correlated Selection (OCS), first introduced in the groundbreaking paper of Fahrbach et al.~\cite{fahrbach2020edge}.
Specifically, the OCSes of \cite{gao2021improved} underlie many rounding-based online algorithms for edge-weighted matching \cite{gao2021improved}, fair  matching \cite{hosseini2024class} and i.i.d matching \cite{tang2022fractional,huang2022power}. 
This powerful machinery can be interpreted as providing oblivious rounding ratios of one for some particular structured solutions (see \cite{buchbinder2023lossless} for discussion). For unrestricted fractional matching, in contrast, this machinery provides useful per-offline-vertex guarantees, but not per-edge oblivious rounding guarantees, as we require. 

\paragraph{Follow-up work.} Following the posting of our paper online, its techniques have found applications in subsequent work. Patel and Wajc \cite{patel2024combinatorial} and subsequently Aouad et al.~\cite{aouad2024adaptive} use our approach of combining proposals from one side of a bipartite graph with contention resolution schemes on the other, for an infinite-horizon Markovian setting. 
Blikstad et al.~\cite{blikstad2025deterministic} similarly used offline CRS and the derandomization approach of \cite{ben1994power} to obtain \emph{deterministic} online bipartite edge coloring algorithms.
Finally, \cite{braverman2025new} 
build on the discount and individual markup technique together with pivotal sampling \cite{srinivasan2001distributions} to improve on our results of \Cref{sec:stochastic} for approximating the optimal stochastic online bipartite matching algorithm (which they refer to as a \emph{philosopher}).
We anticipate further applications of our techniques, in particular the use of offline contention resolution schemes for online algorithm design.

	\vspace{-0.15cm}
\section{Preliminaries}\label{sec:prelims}

%In this section we provide some necessary background, specifically background on joint distributions of random variables (r.v.s) satisfying strong notions of negative correlation.

\textbf{Problem Definition.}
In the online $b$-matching rounding problem, an a priori unknown bipartite graph $G=(L,R,E)$ is revealed online, together with fractions on the edges incident to the arriving edge's vertices that satisfy the $b$-matching constraints. In more detail, 
at each time $t$ an online vertex $t\in L$ arrives, together with fractions $x_{i,t}$ assigned to its edges to \emph{offline neighbors} $i\in R = [n]$.
The fractions $x_{i,t}\in [0,1]$, are promised to satisfy the  (fractional) $b$-matching constraints, namely each vertex $v\in (L\cup R)$ has fractional degree $\sum_{e\ni v} x_e$ at most $b_v$ (the integer $b_v$ is revealed when $v$ is revealed). 
By splitting each online node $t$ with capacity $b_t$ into $b_t$ online unit-capacity nodes with identical edge sets and fractions $x_{it}/b_t$ for each copy of edge $(i,t)$, we can assume WLOG that all online nodes have unit capacity.
Following arrival $t$, we must decide, immediately and irrevocably, which edges of $t$ to add to the output $b$-matching $\calM$, while satisfying the (integral) $b$-matching constraints of matching each vertex $v$ no more than $b_v$ times, and striving for a large oblivious rounding ratio.

Our ODRSes for $b$-matchings repeatedly invoke offline
\emph{contention resolution schemes} (CRSes), whose guarantees for product distributions, due to Feige and Vondr\'{a}k \cite{feige2006allocation}, and arbitrary distributions, due to Bansal and Cohen \cite{bansal2021contention}, are given below. (For completeness, we provide a self-contained proof for the general case in \Cref{appendix:prelims}.) 
\begin{restatable}{lem}{CRS}\label{lem:CRS}
     Let $\calD:2^{[n]}\to \mathbb{R}_{\geq 0}$ be a distribution over subsets of $[n]$, and denote this distribution's support by $\textrm{supp}(\calD):=\{S\subseteq [n] \mid \Pr_{R\sim \calD}[R=S] \neq 0\}$. 
     Then, there exists a randomized algorithm CRS($R,\vec{v})$ which on input set $R\sim \calD$ and vector $\vec{v}\in \mathbb{R}^n$, outputs a subset $O \subseteq R$ of size $|O|\leq 1$ satisfying
     $$\Pr[i\in O] = v_i \cdot \min_{\substack{S\subseteq [n] \\ \sum_{i\in S} v_i\neq 0}}\frac{\Pr[R\cap S\neq \emptyset]}{\sum_{i\in S} v_i} \qquad\qquad  \forall i\in [n].$$
     The algorithm runs in time polynomial in $n$ if $\calD$ is a product distribution. Otherwise, it runs in time $\poly(n,T)$, where $T\geq |\textrm{supp}(\calD)|$ is the time to compute and write down the support of $\calD$.
\end{restatable}

\subsection{Negative association and strong Rayleigh properties}
\begin{Def}\label{def1}
	Random real-valued variables $X_1,\dots,X_n$ are \emph{negatively associated (NA)} if for any two disjoint index sets $I,J\subseteq [n]$, $I\cap J=\emptyset$, and two functions $f:\mathbb{R}^{|I|}\to \mathbb{R}$ and $g:\mathbb{R}^{|J|}\to \mathbb{R}$ both non-decreasing, 
	$$\E[f(X_i : i\in I)\cdot g(X_j : j\in J)]\leq \E[f(X_i : i\in I)]\cdot \E[g(X_j : j\in J)].$$
\end{Def}

By a simple inductive argument, negative association implies negative cylinder dependence~\cite{joag1983negative}.

\begin{lem}\label{na-implies-neg-cylinder}
For any real NA r.v.s $X_1,\dots,X_n$ and reals $x_1,\dots,x_n$, it holds that
	$$\Pr\bigg[\bigwedge_i (X_i\geq x_i)\bigg] \leq \prod_i \Pr[X_i \geq x_i] \qquad \mbox{ \textbf{and} } \qquad \Pr\bigg[\bigwedge_i (X_i\leq x_i)\bigg] \leq \prod_i \Pr[X_i \leq x_i].$$
\end{lem}

The following family of NA distributions is due to Dubhashi and Ranjan \cite{dubhashi1996balls}.
\begin{lem}\label{lem:0-1}
	If $X_1,\dots,X_n$ are binary r.v.s with $\sum_i X_i\leq 1$ always, then they are NA.
\end{lem}

More elaborate NA distributions can be obtained via  simple NA-preserving operations \cite{joag1983negative}.

\begin{lem}\label{na-closure}
NA is closed under products and under disjoint  non-decreasing function composition. That is, if $\vec{X} = (X_1,\dots,X_n)$ is NA, then:
\begin{enumerate}
    \item if $\vec{Y} = (Y_1,\dots,Y_m)$ is NA and $\vec{Y}$ is independent of $\vec{X}$, then so is~$(X_1,\dots,X_n,Y_1,\dots,Y_m)$; 
    \item $(f_1(X_i : i\in I_1),\dots,f_k(X_i : i\in I_k))$ are NA, for any concordant functions (i.e., all non-decreasing or all non-increasing) $f_1,\dots,f_k$ and disjoint sets $I_1,\dots,I_k\subset [n]$, $I_j\cap I_k = \emptyset \,\forall j\neq k$.
\end{enumerate}
\end{lem}

Negative association implies many more useful properties, most notably the applicability of Chernoff-Hoeffding type tail bounds \cite{dubhashi1996balls} and submodular stochastic dominance (see \Cref{appendix:prelims}).

An even stronger negative correlation property than NA is the \emph{strong Rayleigh property (SRP)} (see \Cref{appendix:prelims} for a precise definition). This property implies NA (even under conditioning), as well as concentration of any Lipschitz function \cite{pemantle2014concentration}, and numerous other useful properties.

	\section{Rounding Matchings Online}\label{sec:rounding-matchings}

In this section we mostly focus on the ODRS problem for bipartite matching, deferring a generalization to $b$-matchings to \Cref{app:b-matching-gen}.

We start with a simple $b$-matching algorithm, combining our ODRS for uniform matroids (\Cref{alg:rounding}) with repeated invocations of an \emph{offline} CRS (\Cref{lem:CRS}).

\paragraph{Warm-up: $1-1/e$ rounding ratio.}
We consider the following simple algorithm: Run, for each offline node $i$, an independent instance of the online level-set rounding \Cref{alg:rounding}, applied to $x_{i1},x_{i2},\dots$. When the $i$-th copy of \Cref{alg:rounding} fixes $X_{it} = 1$, we say that $i$ \emph{bids} for online node $t$. Out of  the set of bidders (\textbf{p}otential buyers) $P_t$, we then pick at most one $i\in P_t$ to match $t$ to, using the offline contention resolution scheme of \Cref{lem:CRS}, guaranteeing each node $i$ a probability of $x_{i,t}\cdot \min_{S\subseteq [n]}\frac{\Pr[S\cap P_t\neq \emptyset]}{\sum_{i\in S} x_{i,t}}$ of being matched.
This step runs in polytime, by independence of the bids.
Since \Cref{alg:rounding} is an online algorithm, so is the resultant algorithm. 
Moreover, since the output of \Cref{alg:rounding} satisfies the $b$-matching constraints, and $t$ is matched to at most one neighbor by the CRS, this algorithm's output satisfies the $b$-matching constraints.
Finally, for each offline node set $S\subseteq [n]$, by independence of the bids, the Taylor expansion of $\exp(-x)$ and convexity (see \Cref{convexity}), we have
	\begin{align}\label{eqn:subset-1-1/e}
		\Pr[S\cap P_t\neq \emptyset ] & = 1 - \prod_{i\in S} (1-x_{it}) \geq 1 - \prod_{i\in S} \exp(-x_{it}) \geq \left(1-\frac{1}{e}\right)\cdot \sum_{i\in S} x_{it}.
	\end{align}
Therefore, by \Cref{lem:CRS}, this ODRS matches each edge $(i,t)$ with probability at least $(1-1/e)\cdot x_{i,t}$.

\begin{lem}
    There exists a polytime ODRS for bipartite $b$-matching  with rounding ratio $1-1/e$.
\end{lem}

One natural approach to improve the above algorithm's rounding ratio would be to attempt to negatively correlate the probability of different offline nodes to have previously bid, thus increasing $\Pr[S\cap P_t\neq \emptyset]$ for all offline sets $S\subseteq [n]$. Unfortunately, for a sufficiently large number of offline nodes, (near-)positive correlation is, generally, unavoidable: see \Cref{almost-positive-cylinder-dependence} for a proof of such a Ramsey-theoretic statement.
Instead, to achieve a better rounding ratio, we explicitly negatively correlate \emph{the bids} of previously-non-bidding offline nodes when they do bid.

\subsection{The improved ODRS: Overview and intuition}\label{sec:improved}

In this section we outline our improved ODRS. 
We focus here and in most subsequent sections on simple matchings, for clarity's sake, deferring an extension to $b$-matchings to \Cref{sec:b-matching-extension}.

\paragraph{A bad example, and a false start.} 
Consider a graph with a single online node $t$ with $x_{it}=\frac{1}{n}$ for all $n$ offline nodes. 
In this case, under independent bids, $t$ gets no bids with probability $(1-\frac{1}{n})^n\approx \frac{1}{e}$, and so one of the $n$ offline nodes must get matched with probability no greater than $(1-\frac{1}{e})\cdot \frac{1}{n}$.
For this simple example a rounding ratio of one is possible, by correlating bids for $t$: if we pick exactly one neighbor $i$ to bid with probability $x_{it}$, then each offline node both bids and buys (is matched to) $t$ with probability $x_{it}$.
Unfortunately, this approach is impossible to implement in general: if we let each offline node $i$ bid for at most one time $t$ (recall that this is a simple matching instance) with probability $x_{i,t}$, then, for $s_{it}:=\sum_{t'<t}x_{it'}$ the fractional degree of $i$ at time $t$, we must have $i$ bid for $t$ with probability $p_{it}:=x_{it}/(1-s_{it})$, and yet $\sum_i p_{it}$ may exceed one.

\paragraph{Bin-packing low-degree nodes.}
Our first step is to adopt a \emph{partial} implementation of the above approach, by grouping offline nodes into groups $B$ for which $\sum_{i\in B} p_{it}\leq 1$, and using one $\textrm{Uni}(0,1)$ variable to offer at most one bid per group. 
Specifically, if we denote the set of low-degree neighbors of $t$ by $L_t := \{i\in N(t) \mid s_{it}\leq \theta\}$ for some threshold $\theta\in [0,1]$ that we will optimize later, we will pack low-degree neighbors
into bins using a greedy bin-packing algorithm such that each bin $B$ has $\sum_{i\in B} p_{i,t}\leq 1$. Moreover, each such bin $B$ (but one) has $\sum_{i\in B} p_{i,t}\geq 1/2$ 
and hence has 
high total $x$-value of $\sum_{i\in B}x_{i,t}\geq \frac{1-\theta}{2}$. (Here we use that $p_{i,t}=x_{i,t}/(1-s_{i,t})\leq x_{i,t}/(1-\theta)$.)
Using negative association arguments (see \Cref{sec:prelims}), we can show that this approach essentially replaces multiple bids of small $x$-value with bids of $x$-value equal to the small bids' sum. 
In particular, if the total $x$-value of low-degree neighbors $i\in L_t$ is high, 
this results in $\min_{S\subseteq [n]}\frac{\Pr[S\cap P_t\neq \emptyset]}{\sum_{i\in S}x_{i,t}} \geq 1-\frac{1}{e}+\Omega(1)$, and a rounding ratio greater than $1-\frac{1}{e}$, by \Cref{lem:CRS}.

\paragraph{Group discounting and individual markup.} It remains to address the challenge of increasing $\Pr[S\cap P_t\neq \emptyset]$ if the total $x$-value of low-degree nodes $L_t$ is small, and so very little grouping occurs. 
For this, we note that from the perspective of every offline node $i$, a time step $t$ has: (1) a value (matching probability), and (2) a cost (bidding probability). The latter is indeed a cost, since we do not allow offline nodes to bid more than once. The value at each time is monotone increasing in offline nodes' costs. Now, when we use grouping we decrease the contention faced by the CRS, as it now faces fewer bidding nodes, due to each bin only providing one bid (or none). 
We can therefore safely decrease the bidding probability (cost) in the case that much grouping occurs, giving the grouped nodes a \emph{group discount}, while still obtaining a rounding ratio greater than $1-1/e$.
The upshot of the group discount is that offline nodes now have a higher chance of not having bid before later time steps where they are not grouped, in which case such individual bidders can pay an \emph{individual markup}, allowing us to guarantee that in this case too $\min_{S\subseteq [n]}\frac{\Pr[S\cap P_t\neq \emptyset]}{\sum_{i\in S}x_{i,t}} \geq 1-\frac{1}{e}+\Omega(1)$, with the ensuing $1-\frac{1}{e}+\Omega(1)$ rounding ratio following from the offline CRS of \Cref{lem:CRS}.

\subsection{The core of the improved ODRS}

In this section we introduce our improved ODRS' core subroutine, \Cref{alg:bucketing}.
(The parameter $\theta\in [0,1]$ will be optimized later, and for now it is safe to think of the input vectors $\mathbf{x},\mathbf{v}$ as both equaling the input fractional matching.)
At each time step $t$, the algorithm groups low-fractional-degree offline nodes, and then all remaining nodes, into buckets, using the classic first-fit bin-packing algorithm,\footnote{This algorithm maintains a sorted list of bins with each bin $B$ containing items of overall size at most one, $\sum_{i\in B}s_i\leq 1$, and for each item $i$ in order, adds $i$ to the first (possibly newly-opened) bin that has enough room left.}
with each offline node $i$ having a size of $\frac{x_{i,t}}{1-s_{i,t}}$. This is the conditional probability with which $i$ should bid for $t$ if $i$ is \emph{free}, i.e., has not bid before, to guarantee a marginal bidding probability of $x_{i,t}$. (This follows from our online level-set rounding algorithm's probabilities.)
Items are grouped into bins of size one, and then at most one offline node bids per bin (lines \ref{line:bin-bid-start}-\ref{line:bin-bid-end}), resulting in the set of \emph{candidates}, $C_t$. The free candidates $P_t=C_t\cap F$ then \emph{bid} for $t$, who is then matched to (at most) one bidding neighbor $i\in P_t$, chosen by a CRS.
\begin{algorithm}[t]
	\caption{ODRS-core($\theta,\mathbf{x},\mathbf{v}$)}
	\label{alg:bucketing}
	\begin{algorithmic}[1]
	    \State $\calM\gets \emptyset$ 
	    \State $F\gets [n]$ \Comment{Free offline nodes}
		\For{\textbf{each} time $t$} \label{line:loop-start}
		    \For{\textbf{each} $i\in [n]$}
		    \State $s_{i,t}\gets \sum_{t'<t} x_{i,t'}$
		    \EndFor
		    \State $\calB_t \gets  \textrm{FirstFit}\left(\left\{\left(i,\frac{x_{i,t}}{1-s_{i,t}}\right)\, \middle\vert \, i \in [n], s_{i,t}\leq \theta\right\}\right)$ \Comment{Bucketing low-degree nodes}\label{line:FF} 
		    \State $\calB_t \gets \calB_t  \cup \textrm{FirstFit}\left(\left\{\left(i,\frac{x_{i,t}}{1-s_{i,t}}\right)\, \middle\vert \, i \in [n], s_{i,t}> \theta\right\}\right)$ \Comment{Bucketing high-degree nodes} \label{line:FF-trivial}
		    \State $C_t\gets \emptyset$ \Comment{Candidates}
		    \For{each $B\in \calB_t$} \label{line:bin-bid-start}
		        \State $U\sim \textrm{Uni}[0,1]$
		        \If{$U\leq \sum_{i\in B} \frac{x_{i,t}}{1-s_{i,t}}$}
		        \State $C_t\gets C_t\cup \min_{i} \left\{i\in [k]\, \middle\vert\, U\leq \sum_{\substack{j\in B\\ j\leq i}} \frac{x_{j,t}}{1-s_{j,t}}\right\}$ \label{line:bin-bid-end}
		        \EndIf
		    \EndFor
		    \State $P_t\gets F\cap C_t$ \Comment{Bidders =  free candidates}
		    \State $F \gets F \setminus P_t$  \Comment{Bidders cease being free}
		    \State $O\gets \textrm{CRS}(P_t,\{v_{i,t}\}_i)$ \Comment{Contention Resolution}
		    \If{$|O| = 1$}
		    \State $\calM\gets \calM\cup \{(i,t) \mid i\in O\}$
		    \EndIf
        \EndFor
        \State \textbf{Output} $\calM$
	\end{algorithmic}
\end{algorithm}

First, we note that \Cref{alg:bucketing}---including the bin-packing steps that require $\frac{x_{i,t}}{1-s_{i,t}}\leq 1$---is well-defined provided $\sum_t x_{i,t}\leq 1$ for all $i,t$, since this implies that $x_{i,t}\leq 1-\sum_{t'\leq t} x_{i,t'} = 1-s_{i,t}$.
Next, by construction, each offline node bids (and is thus matched) at most once, while the CRS guarantees that each online node is matched at most once.
To summarize, we have the following.

\begin{obs}\label{bucketing-feasibility}
    \Cref{alg:bucketing} is well-defined and outputs a matching if $\sum_t x_{i,t}\leq 1\,\,\, \forall i\in [n]$.
\end{obs}

We turn to analyzing the algorithm's rounding ratio, starting with the following fact implying that if much grouping occurs, then most bins have high $x$-value. 
\begin{fact}\label{first-fit}
    For each time $t$, at most one bin $B\in \calB_t$ at \Cref{line:FF} has $\sum_{i\in B} x_{i,t} < \frac{1-\theta}{2}.$
\end{fact}
\begin{proof}
    A well-known property of First-Fit is that at most one bin in the output of this algorithm is less than $\frac{1}{2}$ full (see, e.g., \cite{williamson2011design,vazirani2001approximation}).
    Therefore, since $s_{i,t}\leq \theta$ for each offline node $i$ in a bin $B\in \calB_t$ computed in \Cref{line:FF}, each such bin $B$ (barring perhaps one) has 
    \begin{align*}\frac{1}{2} & \leq \sum_{i\in B} \frac{x_{i,t}}{1-s_{i,t}} \leq \sum_{i\in B} \frac{x_{i,t}}{1-\theta}. \qedhere 
    \end{align*}
\end{proof}

The following lemma, combined with \Cref{first-fit}, provides a guarantee no worse than the independent-proposals approach. Indeed, as we will show shortly, it implies a better rounding ratio if much grouping of low-degree neighbors occurs.
\begin{lem}\label{proposal-prob-perturbed-rephrased}
	For any time $t$ and set of offline nodes $S\subseteq[n]$, \Cref{alg:bucketing} satisfies
	$$\Pr[S\cap P_t \neq \emptyset] \geq 1 - \prod_{B\in \calB_t}\left(1-\sum_{i\in B\cap S} x_{it}\right).$$
\end{lem}
\begin{proof}
    Let $F_t$ denote $F$ at time $t$, and let $C_t$ and $P_t = C_t\cap F_t$ be as in \Cref{alg:bucketing}.
    Next, let $C_{i,t} = \mathds{1}[i\in C_t]$ indicate whether $i$ is a candidate at time $t$, and define $F_{i,t}$ and $P_{i,t}$ similarly. 
    By independence of $C_{i,t}$ and $F_{i,t}$ and since $\Pr[C_{i,t}]=\frac{x_{i,t}}{1-s_{i,t}}$ and $\Pr[F_{i,t}]=1-s_{i,t}$ (provable by a simple induction on $t\geq 0$), we have that $\Pr[P_{i,t}] = x_{i,t}$. Therefore, by linearity of expectation and $\sum_{i\in B}C_{i,t}\leq 1$, we have that $\Pr[P_t \cap S\cap B = \emptyset] = 1-\sum_{i\in B\cap S} x_{i,t}$. It remains to show that the events $[P_t \cap S\cap B = \emptyset]$ for different $B$ are negative cylinder dependent, giving the desired inequality, 
    \begin{align*}
    \Pr[P_t\cap S = \emptyset] = \Pr\left[\bigwedge_{B\in \calB_t} (P_t \cap S\cap B = \emptyset)\right] \leq \prod_{B\in \calB_t} \Pr[P_t \cap S\cap B = \emptyset] = \prod_{B\in \calB_t} \left(1-\sum_{i\in B\cap S} x_{i,t}\right).
    \end{align*}
    
    To this end, we will show that the events $[P_t \cap S\cap B = \emptyset]$ are negatively associated (NA), 
    which by \Cref{na-implies-neg-cylinder}
    yields the required above inequality, and hence yields the lemma.    
    First, by \Cref{lem:0-1}, for any bin $B'$ at time $t'<t$, the binary variables $C_{i,t'}$ for all nodes $i\in B'\cap S$, whose sum is at most one, are NA.
    Moreover, the distributions over different bins and different time steps are independent, and so by closure of NA under products (\Cref{na-closure}), all $C_{i,t'}$ are NA.
    Therefore, the variables  $F_{i,t}=1-\max_{t'<t} C_{i,t'}$ are non-increasing functions of disjoint NA variables $C_{i,t'}$, and so by closure of NA under such functions (\Cref{na-closure}), the variables $F_{i,t}$ are NA.
    Moreover, the variables $C_{i,t}$ are independent from the variables $F_{i,t}$ (who are functions only of $C_{i,t'}$ for all $t'<t$), and consequently, by another application of closure of NA under products (\Cref{na-closure}), all variables $\{F_{i,t},C_{i,t} \mid i\in S\}$ are NA.
    Finally, since $\mathds{1}[P_t \cap S\cap B = \emptyset] = 1-\sum_{i\in B\cap S} C_{i,t}\cdot F_{i,t}$ are monotone decreasing functions of disjoint NA variables (different bins $B$), then by closure of NA under such functions (\Cref{na-closure}), the variables $\mathds{1}[P_t \cap S\cap B = \emptyset]$ are NA, 
    as desired.
\end{proof}

\subsection{The improved ODRS}

\Cref{proposal-prob-perturbed-rephrased} and \Cref{first-fit} allow us to argue that if much of $t$'s fractional degree is contributed by neighbors $i$ with low fractional degree, $s_{i,t}\leq \theta$, then $\Pr[S\cap P_t \neq \emptyset]\geq (1-\frac{1}{e}+\Omega(1))\cdot \sum_{i\in S} x_{i,t}$. By \Cref{lem:CRS}, this allows to match edges of such nodes $t$ with probability greater than $(1-\frac{1}{e})\cdot x_{i,t}$. Indeed, we can afford to decrease the fractions $x_{i,t}$ of low-degree nodes (that may be grouped effectively) and still retain a similar $1-\frac{1}{e}+\Omega(1)$ rounding ratio for such $t$. 
In return, we can afford to increase the fractions $x_{i,t}$ of nodes of high degree (who might not be grouped). More precisely, we consider the following transformation to $\mathbf{x}$, capturing this aforementioned \emph{group discount and individual markup}.
For $\theta:=\frac{\delta}{\eps+\delta}$, with $\eps,\delta\in [0,1]$ chosen later, we use the assignment $\hat{\mathbf{x}}:E\to \mathbb{R}_{\geq 0}$:
\begin{align}\label{eqn:x-hat}
    \hat{x}_{i,t} = x_{i,t}\cdot (1-\eps) +  \int_{z=\max(\theta,s_{i,t})}^{\max(\theta,\, s_{i,t}+x_{i,t})} (\eps+\delta) \, dz. 
\end{align}
Intuitively, $\hat{\mathbf{x}}$ decreases/increases the $x$-value of (parts of) edges of $i$ until/after $i$ has fractional degree $\theta$, by a $(1-\eps)$ and $(1+\delta)$ factor, respectively. So, for example, $\hat{x}_{i,t} \geq x_{i,t}\cdot (1-\eps)$ if  $s_{i,t} \leq \theta$
and $\hat{x}_{i,t}=x_{i,t}\cdot (1+\delta)$ if $s_{i,t}> \theta$.
While $\mathbf{\hat{x}}$ is not a fractional matching, since online nodes $t$ may have $\sum_i \hat{x}_{i,t}  > 1$, this assignment does satisfy the fractional degree constraints of \emph{offline} nodes.
\begin{fact}\label{hat-degrees}
    For every offline node $i\in [n]$ and time $t$, we have that $\hat{s}_{i,t}:=\sum_{t'<t}\hat{x}_{i,t'}\leq s_{i,t} \leq 1$.
\end{fact}
\begin{proof}
If $s_{i,t}\leq \theta$, trivially $\hat{s}_{i,t}=s_{i,t}\cdot (1-\eps)\leq s_{i,t}$. Otherwise, by our choice of $\theta = \frac{\delta}{\eps+\delta}$, we have  
\begin{align*}
\hat{s}_{i,t} & 
= s_{i,t}\cdot (1 -\eps) + (s_{i,t}-\theta)^+\cdot(\eps+\delta) = s_{i,t}\cdot (1+\delta)-\delta \leq s_{i,t}.
\qedhere 
\end{align*}
\end{proof}
Our ODRS, given in \Cref{alg:ODRS}, first scales the fractions as in \Cref{eqn:x-hat} (easily computable online, as these $\hat{x}_{i,t}$ only depend on information known by time $t$) and feeds the resultant vectors $\mathbf{\hat{x}},\mathbf{x}$  (playing the roles of $\mathbf{x}$ and $\mathbf{v}$, respectively) into our core ODRS subroutine, \Cref{alg:bucketing}.
By \Cref{hat-degrees} and \Cref{bucketing-feasibility}, we have that \Cref{alg:ODRS} is a well-defined online algorithm, outputting a valid matching.
It remains to analyze this algorithm's rounding ratio, starting with the following lemma.

\begin{algorithm}[t]
	\caption{ODRS($\eps,\delta,\mathbf{x}$)}
	\label{alg:ODRS}
	\begin{algorithmic}[1]
	    \State $\theta \gets \frac{\delta}{\eps+\delta}$
	    \State Init ODRS-core($\theta\cdot (1-\eps),\mathbf{\hat{x}},\mathbf{x}$)
		\For{\textbf{each} time $t$} \label{line:loop-start}
		    \For{\textbf{each} $i\in [n]$}
		    \State $s_{i,t}\gets \sum_{t'<t} x_{i,t'}$
		    \State Compute $\hat{x}_{i,t}$ from $(x_{i,t},s_{i,t},\eps,\delta)$ as in \Cref{eqn:x-hat}
		    \EndFor
		    \State Run next step of ODRS-core($\theta\cdot (1-\eps),\mathbf{\hat{x}},\mathbf{x}$)
		    \EndFor
        \State \textbf{Output} $\calM$ computed by ODRS-core($\theta\cdot (1-\eps),\mathbf{\hat{x}},\mathbf{x}$)
	\end{algorithmic}
\end{algorithm}

\begin{lem}\label{conditional-ratio}
	Let $\epsilon,\delta \in [0,1]$ be such that the function $	f_{\epsilon,\delta}(z) := \exp(-z\cdot (1+\delta)) - (1-z\cdot (1-\eps))$ satisfies:
    \begin{align}
    \label{eqn:conditional-ratio}
    f_{\eps,\delta}(z)\geq 0 \qquad \textrm{for all } z\geq \frac{1-\theta}{2} = \frac{\epsilon}{2(\epsilon+\delta)}.
    \end{align}
    Then \Cref{alg:ODRS} with parameters $\epsilon$ and $\delta$ has rounding ratio at least 
	$$1-\exp\left(-1-\delta+\frac{\eps+\delta}{1-\eps}\right)\cdot\frac{1-\eps}{1+\delta}.$$
\end{lem}
\begin{proof}
    By construction of $\hat{x}_{i,t}$, we have that $\hat{s}_{i,t}\leq \theta\cdot (1-\eps)$ iff $s_{i,t}\leq \theta$. Define
    $L_t := \{i \mid s_{i,t}\leq \theta\} = \{i \mid \hat{s}_{i,t}\leq \theta\cdot (1-\eps)\}$.
    By \Cref{proposal-prob-perturbed-rephrased} and the Taylor expansion of $1-\exp(-z)$, for any offline node set $S\subseteq [n]$ and time $t$, the set of bidders $P_t$ computed by ODRS-core($\theta\cdot (1-\eps),\mathbf{\hat{x}},\mathbf{x}$) satisfies
	\begin{align*}
	\Pr[S\cap P_t \neq \emptyset] & \geq 1 - \prod_{B\not\subseteq L_t} \left(1-\sum_{i\in B\cap S}\hat{x}_{it}\right)\cdot \prod_{B\subseteq L_t} \left(1-\sum_{i\in B\cap S} \hat{x}_{it}\right) \\
	& \geq 1 - \exp\left(-\sum_{i \in S \setminus L_t} x_{it}\cdot (1+\delta)\right)\cdot \prod_{B\subseteq L_t}\left(1-\sum_{i\in B\cap S} x_{it}\cdot (1-\eps)\right).
	\end{align*}
	We now lower bound the ratio $\frac{\Pr[S\cap P_t \neq \emptyset]}{\sum_{i\in S} x_{it}}$. To do so, we note that our lower bound for $\Pr[S\cap P_t \neq \emptyset]$ is a concave function in $\sum_{i\in S} x_{it}$ in the range $[0,1]$, and so our lower bound for the above ratio is minimized by virtually increasing $\sum_{i\in S} x_{it}$ until $\sum_{i\in S} x_{it} = 1$, by adding all $i\notin S$ to the set and possibly adding dummy nodes $J$ with $s_{j,t}=\theta$ for all $j\in J$ and $\sum_J x_{jt}=1-\sum_{i} x_{it}\geq 0$ (see \Cref{convexity}).\footnote{Concretely: adding any value to $\sum_{i\in S} x_{i,t}$ corresponding to adding (part of) another (dummy) element to $S$ decreases the value, so concavity of the resulting function together with \Cref{convexity} completes the proof of this claim.} 
    On the other hand, by \Cref{first-fit}, every bin $B\subseteq L_t$ but at most one bin $B^*$ has
    $\sum_{i\in B}\hat{x}_{i,t}\geq \frac{1-\theta(1-\eps)}{2} \geq \frac{(1-\theta)(1-\eps)}{2}$, and hence $\sum_{i\in B}{x}_{i,t}\geq \frac{1-\theta}{2}$.
    Combining the above, we get
	\begin{align} %\label{proposal/fractions}
	\frac{\Pr[S\cap P_t \neq \emptyset] }{\sum_i x_{it}}
	& \geq 1 - \exp\left(-\left(1-\sum_{i\not\in L_t} x_{it}\right)\cdot (1+\delta)\right)\cdot \prod_{B\subseteq L_t}\left(1-\sum_{i\in B} x_{it}\cdot (1-\eps)\right) \nonumber \\
	& \geq 1 - \exp\left(-\left(1-\sum_{i\in B^*} x_{it}\right)\cdot (1+\delta)\right)\cdot \left(1-\sum_{i\in B^*} x_{it}\cdot (1-\eps)\right) \nonumber \\
    & \geq 1-\exp\left(-1-\delta+\frac{\eps+\delta}{1-\eps}\right)\cdot\frac{1-\eps}{1+\delta}. \label{eqn:halls-ratio}
	\end{align}
    Above, the  second inequality follows from  
    the lemma's hypothesis, in \Cref{eqn:conditional-ratio}, and the last inequality follows from the following claim, whose proof by inspection of $g$ is deferred to \Cref{appendix:b-matching}.

    \begin{restatable}{claim}{minimizer}\label{claim:miminimzer}
        If $\eps,\delta\in [0,1]$, the 
        function $g(y):=1 - \exp\left(-\left(1-y\right)\cdot (1+\delta)\right)\cdot \left(1-y\cdot (1-\eps)\right)$~satisfies
        %, for all $y\in \mathbb{R}$,
        $$\forall y\in \mathbb{R}: \qquad g(y)\geq g\left(\frac{\eps+\delta}{(1-\eps)(1+\delta)}\right)=1-\exp\left(-1-\delta+\frac{\eps+\delta}{1-\eps}\right)\cdot\frac{1-\eps}{1+\delta}.$$
    \end{restatable}
    Finally, the rounding ratio of the ODRS then follows from \Cref{lem:CRS} and \Cref{eqn:halls-ratio}.
\end{proof}

\Cref{conditional-ratio} provides a sufficient condition for $\eps,\delta$ to (possibly) allow for greater-than-$(1-1/e)$ rounding ratios, though this condition requires satisfying infinitely many non-linear constraints.
The following lemma, which follows from convexity of $f_{\eps,\delta}$ and is proven in \Cref{appendix:b-matching}, provides a simpler condition that is sufficient to guarantee all the above infinitely-many constraints.

\begin{restatable}{lem}{sufficient}\label{sufficient-condition}
	Let $f_{\epsilon,\delta}(z)$ be as in \Cref{conditional-ratio}. If $f_{\epsilon,\delta}(a)\geq 0$ and $f'_{\epsilon,\delta}(a)\geq 0$, then $f_{\epsilon,\delta}(b)\geq 0$ $\forall b\geq a$.
\end{restatable}

Combining the above, we are finally ready to optimize the rounding ratio of \Cref{alg:ODRS}.

\begin{thm}\label{thm:ODRS}
    \Cref{alg:ODRS} run with $(\eps,\delta)\approx (0.0480,0.0643)$ achieves a rounding ratio of $\ratio$.
\end{thm}
\begin{proof}
 By lemmas \ref{conditional-ratio} and \ref{sufficient-condition}, the rounding ratio of \Cref{alg:bucketing} with parameters $(\eps,\delta)$ is at least as high as $1-\exp\left(-1-\delta+\frac{\eps+\delta}{1-\eps}\right)\cdot\frac{1-\eps}{1+\delta}$, provided that $f_{\eps,\delta}\left(\frac{1-\theta}{2}\right)\geq 0$ and $f'_{\eps,\delta}\left(\frac{1-\theta}{2}\right)\geq 0$, where $\frac{1-\theta}{2}=\frac{\eps}{2(\eps+\delta)}$ and $f_{\eps,\delta}$ is as in \Cref{conditional-ratio}.
 Thus, the algorithm's rounding ratio is at least as high as 
 the value of the following program, when run with parameters given by its optimal solution $(\eps,\delta)$.
 Off-the-shelf numerical solvers yield a maximum value of $\alpha=\ratio$ at $(\eps,\delta)\approx (0.0480,0.0643)$.
\begin{align*}
	\max \quad & \alpha = 1-\exp\left(-1-\delta+\frac{\eps+\delta}{1-\eps}\right)\cdot\frac{1-\eps}{1+\delta} \\
	\textrm{s.t.} \quad & \exp\left(-\frac{\epsilon}{2(\epsilon+\delta)}\cdot (1+\delta)\right) - 1 + \frac{\eps}{2(\eps+\delta)}\cdot (1-\epsilon) \geq 0 \\
	& -(1+\delta)\cdot \exp\left(-\frac{\epsilon}{2(\epsilon+\delta)}\cdot (1+\delta)\right) + 1-\epsilon \geq 0 \\
	& \epsilon,\delta\in [0,1]. \qedhere
\end{align*}
\end{proof}

% MATHEMATICA CODE:

% x[eps_, delta_] := eps/(2*(eps + delta))
% NMaximize[{1 - (1 - eps)/(1 + delta)*
%     Exp[-(1 + delta) + ((eps + delta)/(1 - eps))], 
%   Exp[-x[eps, delta]*(1 + delta)] - 1 + x[eps, delta]*(1 - eps) >= 
%    0, -(1 + delta)*Exp[-x[eps, delta]*(1 + delta)] + 1 - eps >= 0,
%   0 <= eps <= 1, 0 <= delta <= 1}, {eps, delta}]
 
%  Expected output:
% {0.652808, {eps -> 0.0480696, delta -> 0.0643458}}

\noindent\textbf{Computational considerations.} The algorithm implied by \Cref{thm:ODRS}, as stated, takes exponential time, due to its reliance on the CRS of \Cref{lem:CRS} for non-independent bid distributions. 
As we show in \Cref{sec:polytime}, a slight downscaling of the $x$-values allows us to achieve essentially the same rounding ratio, 
while also guaranteeing that the number of bidders at each time step is small. This allows us to compute the support of the distribution of $P_t$ in polynomial time, which by \Cref{lem:CRS} then gives us the following polytime counterpart to \Cref{thm:ODRS}.

\begin{restatable}{thm}{polymatching}\label{1-1/e+delta-rounding-efficient}
	For any $\gamma>0$, there exists an $n^{O(1/\gamma)}$-time bipartite matching ODRS with rounding ratio of $\ratio-\gamma$.
\end{restatable}

    \section{Online Level-Set Rounding}\label{sec:level-sets}

In this section we present two online level-set rounding algorithms.
The guarantees of such algorithms, first obtained by \cite{srinivasan2001distributions,DBLP:conf/soda/Srinivasan07} in \emph{offline settings}, are as follows.

On input $\vec{x} \in [0,1]^n$ with integer sum $\sum_i x_i \in \mathbb{Z}$ (possibly by adding a dummy value) and partial sums $s_j := \sum_{i\leq j}x_j$, a level-set rounding algorithm must output a set $\calS \subseteq [n]$ with characteristic vector $\vec{X}$ given by $X_i:=\mathds{1}[i\in \calS]$ with partial sums $S_j := \sum_{i\leq j} X_i = |\calS \cap [j]|$ satisfying the following properties for all $i$. 
\begin{enumerate}[label=(P{{\arabic*}})]
	\item \textbf{(Marginals)} $\E[X_i] = x_i.$ \label{prop:marginals}
	\item \textbf{(Rounding)} $S_i \in \{\lfloor s_i \rfloor,\lceil s_i \rceil\}$. \label{prop:floor/ceil}
\end{enumerate}
Property \ref{prop:marginals} and linearity of expectation imply that $\E[S_i] = s_i$. Property \ref{prop:floor/ceil} requires that $S_i$ always equals this expectation, ``up to rounding''. 

In this work, we require \emph{online} counterparts to the above. Here, 
the input is again a sequence of reals, $x_1,x_2,\dots \in [0,1]$, but now they are revealed in an online fashion. An online level-set rounding algorithm maintains a monotonically increasing set $\calS \subseteq \mathbb{N}$ (initially empty), and decides at time $t$ (when $x_t$ is revealed) whether to add $t$ to its output set $\calS$, immediately and irrevocably. 

As we shall see shortly, obtaining these first two properties is easy to achieve online. 
However, for several of our applications in \Cref{sec:applications-body} it will be more useful to additionally guarantee strong negative correlation between the elements of $\vec{X}$. Specifically, we require the following.
\begin{enumerate}[label=(P{{\arabic*}})]
	  \setcounter{enumi}{2}
	\item \textbf{(Strongly Rayleigh)} The joint distribution $\vec{X}$ is strongly Rayleigh (and thus, also NA). \label{prop:sr}
\end{enumerate}

\paragraph{A warm-up:} We start with a simple algorithm that satisfies Properties \ref{prop:marginals} and \ref{prop:floor/ceil}, but not Property \ref{prop:sr}.\footnote{We thank an anonymous reviewer for this suggestion.}
Upon initialization, we draw a uniform threshold $\tau\sim\textrm{Uni}[0,1]$. Next, for every time step $j$, for $s_j:=\sum_{i\leq j} x_{i}$ the running sum until the $j^{th}$ arrival (inclusive), we add element $j$ to the output set $\calS$ if the interval $(s_{j-1},s_j]$ contains a point $n+\tau$ for some integer $n$.
Thus, $\Pr[i\in \calS] = \Pr[(s_{j-1},s_j] \cap (\mathbb{N}+\tau)] = s_j - s_{j-1} = x_i$, and this algorithm satisfies Property \ref{prop:marginals}.
Moreover, as we output one element for any sum interval $(n,n+1]$, it is easy to see that the output by time $j$ contains either $\lfloor s_j\rfloor$ or $\lceil s_j\rceil$ many elements, and so this algorithm satisfies Property \ref{prop:floor/ceil}.
To see that this algorithm does not satisfy \ref{prop:sr}, which implies pairwise negative correlation, we note that for an input $\vec{x}=\frac{1}{2}\cdot \vec{1}$ (i.e., $x_i=1/2$ for all $i$), negative correlation is (very much) violated, since $X_i=X_{i+2}$ for all $i$. 
\medskip

The main contribution of this section is an online level-set rounding algorithm satisfying both basic properties of an online level-set algorithm, and the strong negative correlation Property~\ref{prop:sr} (necessary for some of our applications in \Cref{sec:applications-body}). As this algorithm's analysis involves coupling with the \emph{offline} level-set rounding algorithm of \citet{srinivasan2001distributions,DBLP:conf/soda/Srinivasan07}, we first start by revisiting this offline algorithm and its properties.

\subsection{Offline level-set rounding}
Here we consider an \emph{offline} procedure due to \cite{srinivasan2001distributions}. 
Recall that our objective will be to design an \emph{online} algorithm satisfying the aforementioned properties, but for now we revisit the offline algorithm of \cite{srinivasan2001distributions} and prove that when run with a certain specific way of choosing the key indices $i_1$ and $i_2$ (see \Cref{alg:linear-srinivasan-rounding}), then it satisfies the above three properties.  

As usual, ``$\min(S)$" in Algorithm~\ref{alg:linear-srinivasan-rounding} refers to the minimum element of a nonempty set $S$; the index $i_2$ in Algorithm~\ref{alg:linear-srinivasan-rounding} is always well-defined since $\sum_i x_i \in \mathbb{Z}$ and, as we shall see, $\sum_i y_i = \sum_i x_i$ always. We emphasize that the 
algorithm of \cite{srinivasan2001distributions} allows much liberty in the choice of $i_1$ and $i_2$, and that our specific choice of $i_1$ and $i_2$ is required to satisfy Property~\ref{prop:floor/ceil}. Also, the random choice in each call to STEP is independent of the random choices of the past.

\noindent\begin{minipage}[t]{0.46\textwidth}
	\vspace{-0.5cm}
	\begin{algorithm}[H]
		\caption{Offline Level-Set Rounding}
		\label{alg:linear-srinivasan-rounding}
		\begin{algorithmic}[1]
			\State Initialize $\vec{y} \gets \vec{x}$
			\While{$frac(y) := \{i \mid y_i \not\in\{0,1\}\}\neq \emptyset$}
			\State Let $i_1\gets  \min (frac(y))$
			\State Let $i_2\gets \min (frac(y)\setminus \{i_1\})$
			\State STEP$(y_{i_1},y_{i_2})$
			\EndWhile
			%		\If{$|frac(Y)|=1$}
			%		\State Let $i\in frac(y)$
			%		\State Set $y_i\sim \Ber(y_i)$ \label{line:last}
			%		\EndIf
			\State Output $\calS:= \{i \mid y_i = 1\}$
		\end{algorithmic}
	\end{algorithm}	
\end{minipage}
\hfill
\begin{minipage}[t]{0.5\textwidth}
	\vspace{-0.5cm}	
	\begin{algorithm}[H]
		\caption{STEP$(A,B) \quad \vartriangleright$ modifies inputs}
		\label{alg:offline-rounding-step}
		\begin{algorithmic}[1]
			\If{$A+B < 1$}
			\State $(A,B) \gets \begin{cases}(A+B,0) & \textrm{w.p. } \frac{A}{A+B} \\
			(0,A+B) & \textrm{w.p. } \frac{B}{A+B}
			\end{cases}$
			\Else \Comment{$A+B \in [1,2)$}
			\State $(A,B) \gets \begin{cases}(1,A+B-1) & \textrm{w.p. } \frac{1-B}{2-A-B} \\
			(A+B-1,1) & \textrm{w.p. } \frac{1-A}{2-A-B}
			\end{cases}$
			\EndIf 
		\end{algorithmic}
	\end{algorithm}	
\end{minipage}

\vspace{0.2cm}
That \Cref{alg:linear-srinivasan-rounding} terminates and outputs a random set satisfying Property \ref{prop:marginals} is known \cite{srinivasan2001distributions}. Nonetheless, for completeness, we provide a proof of this fact in \Cref{appendix:level-sets}.
In the same appendix we prove by induction on the number of invocations of STEP() that every prefix sum is preserved with probability one, up to rounding, and so \Cref{alg:linear-srinivasan-rounding} satisfies Property \ref{prop:floor/ceil}.

\begin{restatable}{lem}{offlineprefix}
	Let $s_j := \sum_{i\leq j} x_i$ and $Y^t_j := \sum_{i\leq j} y^t_i$, where $y^t_i$ is the value $y_i$ after $t$ applications of STEP in 	\Cref{alg:linear-srinivasan-rounding}. Then, $Y^t_j \in \left[\lfloor s_j \rfloor, \lceil s_j \rceil \right]$ for all $j$ and $t$.
\end{restatable}
Strong negative correlation (Property \ref{prop:sr}) was proven in \cite{branden2012negative} for any instantiation of \cite{srinivasan2001distributions}, as long as we use some fixed choice of $i_1$ and $i_2$, as we do in \Cref{alg:linear-srinivasan-rounding}: 

\begin{lem}
	Let $X_i := \mathds{1}[i\in \calS]$ be the indicators for $i$ being in the output set of \Cref{alg:linear-srinivasan-rounding}. Then, $\vec{X}=(X_1,X_2,\dots,X_n)$ is a family of strongly Rayleigh random variables.
\end{lem}

To conclude, we have the following.
\begin{lem}\label{offline-props}
	\Cref{alg:linear-srinivasan-rounding} satisfies properties \ref{prop:marginals} \ref{prop:floor/ceil} and \ref{prop:sr}.
\end{lem}

\subsection{The online level-set rounding algorithm}\label{sec:main-round}

We now turn to providing an online counterpart to \Cref{alg:linear-srinivasan-rounding}. Here, 
the input is a sequence of reals, $x_1,x_2,\dots \in [0,1]$, revealed in an online fashion. \Cref{alg:rounding} maintains a set $\calS \subseteq \mathbb{N}$ (initially empty), and decides at time $t$ (when $x_t$ is revealed) whether to add $t$ to its output set $\calS$, immediately and irrevocably. 
Let $X_t := \mathds{1}[t\in \calS]$, and let $S_t := \sum_{t'\leq t} X_{t'}$ denote the cardinality of $\calS$ right after time $t$, i.e., $S_t=|\calS\cap [t]|$. Finally, let $s_t := \sum_{t'\leq t} x_{t'}$. We will show that \Cref{alg:rounding} satisfies properties \ref{prop:marginals}, \ref{prop:floor/ceil} and \ref{prop:sr} given by its offline counterpart, \Cref{alg:linear-srinivasan-rounding}.

\begin{algorithm}[h]
	\caption{Online Level-Set Rounding}
	\label{alg:rounding}
	\begin{algorithmic}[1]
		\State Initialize $\calS\gets \emptyset$
		\For{arrival $x_t$ (at time $t\geq 1$)}
		\State add $t$ to $\calS$ with probability 
		$p_t := \begin{cases}
			0 & |\calS| = \lceil s_t \rceil \\ 
			1 & |\calS| < \lfloor s_t \rfloor \\
			\frac{x_t}{\lfloor s_{t-1} \rfloor + 1 - s_{t-1}}  & |\calS| = \lfloor s_{t} \rfloor = \lfloor s_{t-1} \rfloor \\
			\frac{s_t - \lfloor s_t \rfloor}{s_{t-1} - \lfloor s_{t-1} \rfloor}  & |\calS| = \lfloor s_{t} \rfloor > \lfloor s_{t-1} \rfloor \textrm{ and } s_{t-1}\neq \lfloor s_{t-1}\rfloor \\
			0 & \textrm{else.}
		\end{cases}$
		\EndFor
		\State Output $\calS$
	\end{algorithmic}
\end{algorithm}	
The random bits used to make the random choice in each step of \Cref{alg:rounding} are independent of the random bits used in the past: i.e., while $p_t$ depends on $S_{t-1}$, the random bits used to generate the event of probability $p_t$ in step $t$, are independent of the random bits of past steps. 

\subsection{Coupling the online and offline algorithms}
Here we prove that \Cref{alg:rounding} satisfies properties \ref{prop:marginals}, \ref{prop:floor/ceil} and \ref{prop:sr}.
Proving properties \ref{prop:marginals}, \ref{prop:floor/ceil} directly is not too hard (see \Cref{level-set:simple}).\footnote{
In fact, \Cref{alg:rounding} can be shown to be the \emph{only} online algorithm satisfying both these simple properties if we further restrict it to only storing in memory the partial sums $s_{t-1}$ and $S_{t-1}$ at the beginning of step $t$.}
Proving Property \ref{prop:sr}, on the other hand, is slightly more involved; we do so now, using that the much-simpler properties
\ref{prop:marginals} and \ref{prop:floor/ceil}
hold for both the offline and online algorithms. %Again, both properties are proven for the online algorithm in \Cref{appendix:level-sets}. 
To prove that the last property holds, we show that both algorithms induce the \emph{same distribution} over outputs. We outline a proof here, deferring a complete proof to \Cref{appendix:level-sets}.

\begin{restatable}{thm}{equivalence}\label{lem:reduction}
Fix a vector $\vec{x}\in [0,1]^n$. Let $\calD$ and $\calD'$ denote the probability distributions on $\{0,1\}^n$ induced by the online and offline algorithms run on $\vec{x}$, respectively.
Then, $\calD=\calD'$.
\end{restatable}

\begin{proof}[Proof (Sketch)]
We prove by induction on $t\in [n]$ that the following holds:
\begin{equation}
    \label{eqn:online-equals-offline-body}
    \forall (b_1, \ldots, b_{t-1}) \in \{0,1\}^{t-1}, ~
\Pr_\calD[X_t = 1 \bigm| (\forall i < t, ~X_i = b_i)] =
\Pr_{\calD'}[X_t = 1 \bigm| (\forall i < t, ~X_i = b_i)]. 
\end{equation}
The (strong) inductive hypothesis, whereby \Cref{eqn:online-equals-offline-body} holds for all $t'< t$, clearly implies that for all $(b_1,\dots,b_t)\in \{0,1\}^{t-1}$ we have that $\Pr_\calD[(\forall i<t, ~X_i=b_i)]=\Pr_{\calD'}[(\forall i<t, ~X_i=b_i)]$.
Similarly, \Cref{eqn:online-equals-offline-body} holds for all $t\leq n$ iff the two distributions are equal. The proof then proceeds by considering the different possibilities for the prefix sums $s_{t-1}$ and $s_t$, which for the online algorithm directly determine the above conditional probability. These conditions also determine the offline algorithm's performance up to the $(t-1)^{st}$ invocation of STEP(), from which a careful analysis implies that the conditional probabilities of the offline and online algorithms are equal.
%See \Cref{appendix:level-sets} for a full proof.
\end{proof}

\Cref{lem:reduction} combined with \Cref{offline-props} immediately implies the desired properties of our \Cref{alg:rounding}, as these properties are determined by the distribution over output sets.

\begin{thm}\label{alg:online-rounding-properties}
	\Cref{alg:rounding} is an online algorithm satisfying properties \ref{prop:marginals}, \ref{prop:floor/ceil} and \ref{prop:sr}.
\end{thm}

	\section{Applications}\label{sec:applications-body}

In this section we provide applications and extensions of our online rounding algorithms. Mirroring the exposition in \Cref{sec:applications}, we start with the application of our online ($b$-)matching algorithms, omitting the edge-weighted matching application with ML predictions, which follows directly from our main result.

\subsection{Applications and extensions of the matching ODRS}

\subsubsection{Online edge coloring (multigraphs)}\label{sec:app-edge-coloring}

We show here how to obtain the first online algorithm that achieves a better-than-$2$ competitive factor for online edge-coloring in bipartite {\em multigraphs} under adversarial arrivals, specifically under one-sided node arrivals, as studied in \cite{cohen2019tight}. We assume that upon arrival of each node, its neighboring edges are revealed and must be colored.
As already mentioned, for online edge coloring, it is known that the greedy algorithm's competitive ratio of $2$ is tight for graphs of low maximum degree $\Delta=O(\log n)$ \cite{bar1992greedy}, while better competitive ratios are known for high-degree graphs under various arrival models, all for simple graphs
\cite{cohen2019tight,saberi2021greedy,kulkarni2022online,bhattacharya2021online,bahmani2012online,blikstad2024simple,blikstad2024online}, with the exception of the random-order result of \cite{aggarwal2003switch}.

We obtain our improved bounds for multigraphs by extending a reduction from online edge coloring to ``fair'' matching, suggested by  \cite{cohen2019tight} for simple graphs (see also \cite[Appendix B]{saberi2021greedy}), given in \Cref{reduction}.

\begin{Def}
We say that a random matching $\calM$ in a multigraph with maximum degree $\Delta$ is \emph{$\alpha$-fair} if for each (parallel) edge $e$ we have that $\Pr[e\in \calM] \geq \frac{1}{\alpha\Delta}$.
\end{Def}
An $\alpha$-fair algorithm with $\alpha=1$ gives the highest uniform lower bound on edges' matching probability. An $\alpha$-fair algorithm with other $\alpha$ ($\alpha > 1$) therefore $\alpha$-approximates this maximally-fair algorithm.

With the above terminology, we can now state the reduction.

\begin{lem}\label{reduction}
An online $\alpha$-fair matching algorithm yields an online $(\alpha+o(1))\Delta$-edge-coloring algorithm for $n$-node bipartite multigraphs with maximum degree $\Delta=\omega(\log n)$.
\end{lem}

We defer a brief proof sketch of the above to the end of the section, and turn to discussing its consequences and converses.

\begin{cor}\label{reduction-ODRS}
An ODRS with rounding ratio $1/\alpha$ yields an online $(\alpha+o(1))\Delta$-edge-coloring algorithm for $n$-node multigraphs with maximum degree $\Delta=\omega(\log n)$.
\end{cor}
\begin{proof}
Consider the fractional matching  $\mathbf{x}$ assigning value $x_e = \frac{\kappa(e)}{\Delta}$ to the simple edge $e$ with $\kappa(e)$ parallel copies in the multigraph. (Note that $\sum_{e\ni v} \kappa(e)\leq \Delta$ by definition, so this is indeed a fractional matching.) Applying an ODRS with rounding ratio $1/\alpha$ to this fractional matching $\mathbf{x}$ yields a randomized matching $\calM$ that matches every (parallel) edge in the graph with probability $\frac{1}{\alpha \Delta}$, i.e., it yields an $\alpha$-fair matching.
\end{proof}

We note that there is in fact an equivalence between $\alpha$-fair matching algorithms and $\alpha$-competitive edge-coloring algorithms, as ``the converse'' also holds.
\begin{obs}
An online $\alpha\Delta$-edge-coloring algorithm $\calA$ for $n$-node multigraphs with maximum degree $\Delta$ yields an ODRS with rounding ratio $1/\alpha$ for fractional matchings $\mathbf{x}$ with $x_e\cdot \Delta$ integral for all $e$.
\end{obs}
\begin{proof}
Given a fractional matching $\mathbf{x}\in \mathbb{Q}^E$ in a simple graph $G=(V,E)$, using that $x_e\cdot \Delta$ is integral for each edge $e$, we provide to $\calA$ a multigraph $\calH = (V,F)$ with each edge $e$ having multiplicity $\kappa(e) = x_e\cdot \Delta$ in $\calH$, 
and then randomly sample one of the $\alpha\Delta$ edge colors (matchings) $\calM$ computed by $\calA$.
This yields an ODRS with rounding ratio of $1/\alpha$, since for each edge $e\in E$ we have that 
\begin{align*}
\Pr[e\in \calM] & = \kappa(e)\cdot \frac{1}{\alpha\Delta} = \frac{x_{e}}{\alpha}. \qedhere 
\end{align*}
\end{proof}

By the equivalence between edge coloring multigraphs and ODRSes for matchings, together with our (upper and lower) bounds on the rounding ratio of such ODRSes, we obtain the following bounds on the number of colors needed to edge-color multigraphs online.

\begin{thm}
In $n$-node bipartite multigraphs with maximum degree $\Delta$ under adversarial one-sided vertex arrivals, there exists an online $1.533\Delta$-edge-coloring algorithm, assuming $\Delta=\omega(\log n)$. In contrast, no $(1/(2\sqrt{2}-2+\eps))\Delta\approx (1.207-O(\eps)) \Delta$-edge-coloring exists even for $\Delta=2$.
\end{thm}

We note in passing that all known positive results for online edge coloring under adversarial arrivals follow from \Cref{reduction}, and indeed from \Cref{reduction-ODRS} \cite{cohen2019tight,saberi2021greedy,kulkarni2022online}. The reason most of these results do not extend to multigraphs is that their ODRSes required their input fractional matchings $\mathbf{x}$ to be of the form $x_e = \frac{1}{\Delta}$, or more generally for $x_e = o(1)$ for each edge $e$ (including its multiplicities). 
The work of \citet{saberi2021greedy}
is a lone exception to this rule, and combining its ODRS with \Cref{reduction} also yields an algorithm for multigraphs, but using as many as $1.9\Delta$ colors.

We conclude the section with a brief proof sketch of \Cref{reduction}.

\begin{proof}[Proof sketch of \Cref{reduction}]
We describe the algorithm in an offline setting first.
Let $C$ be such that $C=\omega(\log n)$ and $C=o(\Delta)$. (Such $C$ exist, by our necessary assumption that $\Delta=\omega(\log n)$.)
Initially, the entire multigraph is uncolored.
For some $\Delta/C$ many rounds, compute $\alpha \cdot C$ colors as follows:
In the uncolored subgraph $U$, compute $\alpha\cdot C$ many $\alpha$-fair matchings $\calM_1,\dots,\calM_{\alpha C}$, and let these occupy the next $\alpha C$ colors.
Since every vertex has at most $\alpha C = o(\Delta)$ many edges colored during a round, any vertex with degree close to $\Delta$ at the beginning of the round has degree $\Delta-o(\Delta)$ by the end of the round.
Therefore, since these $\calM_i$ are $\alpha$-fair matchings, and so match each (parallel) edge with probability at least $\frac{1}{\alpha\cdot\Delta}$, by standard concentration bounds, all vertices that have degree close to $\Delta$ have some $C \cdot (1-o(1))$ many edges colored during a round (with high probability).  This allows us to compute a high-probability upper bound of on the resultant uncolored multigraph at the end of the round (possibly useful to compute $\alpha$-fair matchings), and moreover implies that after $\Delta/C$ many rounds (and so with $(\Delta/C)\cdot \alpha\cdot C = \alpha\cdot \Delta$ colors used), the uncolored multigraph has maximum degree $\Delta - (\Delta/C)\cdot C \cdot (1-o(1)) = o(\Delta)$, and can be colored greedily using a further $2\cdot o(\Delta)=o(\Delta)$ many colors, for $(\alpha+o(1))\Delta$ colors overall.

While the above algorithm was described in an offline setting, if one can compute an $\alpha$-fair algorithm online (possibly with information regarding the high-probability upper bound on $\Delta$ in each round), then the entire algorithm can be made to run online: 
for each (vertex/edge) arrival, perform the next steps of the $\alpha \Delta$ many $\alpha$-fair matching algorithms outputting color classes $\calM_1,\dots,\calM_{\alpha\Delta}$, where for the $i$-th such algorithm, we simulate only the arrival of the parts of the subgraph that are not contained $\calM_1,\dots,\calM_{i-1}$. The same approach is used to simulate the final greedy coloring steps, using the fact that the greedy $(2\Delta-1)$-edge-coloring algorithm is an online algorithm.
\end{proof}

\color{black}

\subsubsection{Stochastic extension}\label{sec:stochastic}
In this section, we generalize our matching ODRS and its analysis to the following stochastic online bipartite matching problem. 
At each time $t$, an online node arrives with weight vector $\vec{w}_t\sim \calD_t$, where $\calD_1,\calD_2,\dots$ are a priori known independent distributions. 
A simple example which we focus on for notational simplicity (though our approach generalizes to this problem in full generality) is as follows: each online node $t$ arrives according to a Bernoulli event $A_t\sim \Ber(p_t)$, each edge $(i,t)$ having some intrinsic value $w_{it}$, but realized weight $v_{it} = w_{it}\cdot A_t$.
This problem (in its full generality) is an online Bayesian selection problem, and a number of algorithms with competitive ratio of $\frac{1}{2}$ are known for it \cite{ezra2020online,dutting2020prophet,feldman2015combinatorial}. By the classic lower bound of Krengel and Sucheston for the single-item prophet inequality problem \cite{krengel1978semiamarts}, the above ratio is best-possible when comparing with the offline optimum algorithm.

In \cite{papadimitriou2021online}, Papadimitriou et al.~initiate the study of this problem in terms of the (polytime) approximability of the optimum \emph{online} algorithm.
Perhaps surprisingly, they show that while for simple cases of the problem (e.g., the single-offline-node problem), the optimum online algorithm is computable using a simple poly-sized dynamic program, in general it is \textsc{pspace}-hard to approximate the optimum online algorithm within some absolute constant $\alpha<1$ (say, $\alpha\approx 0.999999$), and showed that a better-than-$0.5$ approximation of the optimum online algorithm can be obtained by polytime online algorithms. This positive result was subsequently improved to $0.526$ implicitly \cite{saberi2021greedy} and to $1-1/e\approx 0.632$ explicitly \cite{braverman2022max}. 
All these results are achieved by rounding the following LP, which can be shown to upper bound the optimum online algorithm's value \cite{papadimitriou2021online,torrico2022dynamic}.
\begin{alignat}{5}
	\max\quad  & \sum_{i,t} w_{i,t} \cdot x_{i,t} \tag{LP $OPT_{on}$}\label{LP} \\
	\textrm{s.t.} \quad & \sum_{t,r} x_{i,t} \leq 1  & \forall i\in [n] \\
	& \sum_{i} x_{i,t} \leq p_{t} & \forall t \label{LP:flow-for-t}\\
	& x_{i,t} \leq p_{t} \cdot \left(1-\sum_{t'<t} x_{i,t'}\right) & \forall i,t \label{OPTon-constraint} \\
	& x_{i,t} \geq 0 & \forall i,t.
\end{alignat}

\begin{lem}\label{lem:LP-bound}
	The optimum of the above LP upper bounds the expected gain of any online algorithm.
\end{lem}

Most of these constraints are evident, and follow from (expected) matching constraints, which apply also to offline algorithms.
The one constraint which applies to online algorithms but not to offline ones is Constraint \eqref{OPTon-constraint}. This constraint follows from the fact that if offline node $i$ is matched to $t$, then  $t$ has already arrived, and that $i$ is unmatched before time $t$ --- two events which are independent for online algorithms.

We note that in their paper's full version, \citet{papadimitriou2021online} (private communication) show that the above LP has a ``gap'' of $1-\frac{1}{2e}$, in the sense that no online algorithm (even computationally-unbounded ones) can get a better-than-$(1-\frac{1}{2e})$ fraction of this LP's value. So, an approximation ratio greater than $1-\frac{1}{2e}$ is unachievable using this LP, and a ratio of $1-\frac{1}{e}$ is achievable \cite{braverman2022max}. Is the latter bound tight? In what follows, who show how the approach underlying our ODRSes, combined with (and simplifying) the high-level analysis of \cite{braverman2022max}, allows one to break this bound.

\paragraph{ODRS for the stochastic setting.}
As in the non-stochastic setting of \Cref{sec:rounding-matchings}, specifically \Cref{alg:bucketing}, we group offline nodes and have them coordinate a single bid per group (bin). Similarly, our stochastic counterpart to \Cref{alg:ODRS} applies group discounting and individual markup to the target marginal bidding probability, and then calls the (soon-to-be) modified \Cref{alg:bucketing} 
with parameter $\theta\cdot (1-\eps)$ and vector $\hat{x}_{i,t}$ given by \Cref{eqn:x-hat}, restated below for ease of reference.
\begin{align*}
    \hat{x}_{i,t} = x_{i,t}\cdot (1-\eps) +  \int_{z=\max(\theta,s_{i,t})}^{\max(\theta,\, s_{i,t}+x_{i,t})} (\eps+\delta) \, dz. 
\end{align*}
The three main differences compared to \Cref{sec:rounding-matchings} are as follows: 
First, to account for the probability $p_t$ of $t$ to arrive, the conditional probability for $i$ to bid if $i$ is free and $t$ arrives used in \Cref{line:FF}-\Cref{line:bin-bid-end} is now $\frac{\hat{x}_{i,t}}{p_t\cdot (1-\hat{s}_{i,t})}$, for $\hat{s}_{i,t}:=\sum_{t'<t}\hat{x}_{i,t}$. 
Second, we allow offline nodes to bid more than once, and denote them as \emph{free} (i.e., belonging to $F$) as long as they are not matched.
This avoids offline nodes becoming positively correlated due to an unlikely online node's arrival.
The offline nodes' matching constraints are trivially satisfied.
Finally, satisfying the online nodes' matching constraints is even easier than in \Cref{alg:bucketing} for the non-stochastic setting. Rather than using a CRS, we have an arriving online node $t$ that receives bids simply greedily match to the bidding node $i$ maximizing the value $w_{i,t}$.
See the pseudocode in \Cref{alg:stochastic}.

\begin{algorithm}[h]
	\caption{Stochastic Matching Algorithm}
	\label{alg:stochastic}
	\begin{algorithmic}[1]
	    \State $\calM\gets \emptyset$ 
	    \State $\mathbf{x}\gets$ solution to \eqref{LP}
  \For{\textbf{each} time $t$} \label{line:loop-start-stochastic}
		    \For{\textbf{each} $i\in [n]$}
		    \State compute $\hat{x}_{i,t}$ as in \Cref{eqn:x-hat}, and $\hat{s}_{i,t}=\sum_{t'<t} \hat{x}_{i,t}$. 
		    \EndFor
		    \State $\calB_t \gets \textrm{FirstFit}\left(\left\{\left(i,\frac{\hat{x}_{i,t}}{p_t\cdot (1-\hat{s}_{i,t})}\right)\, \middle\vert \, i \in [n], s_{i,t}\leq  \theta\right\}\right)$ \Comment{Bucketing low-degree nodes} \label{line:FF-stochastic}
            \State $\calB_t \gets \calB_t\cup \left\{\{i\}\, \middle\vert \, i \in [n], s_{i,t}> \theta\right\}$ \Comment{Trivially bucketing high-degree nodes}\label{line:FF-trivial-stochastic} 
		    \State $C_t\gets \emptyset$ \Comment{Candidates}
		    \For{each $B\in \calB_t$} \label{line:bin-bid-start}
		        \State $U\sim \textrm{Uni}[0,1]$
		        \If{$U\leq \sum_{i\in B} \frac{\hat{x}_{i,t}}{p_t\cdot (1-\hat{s}_{i,t})}$}
		        \State $C_t\gets C_t\cup \min_{i} \left\{i\in [k]\, \middle\vert\, U\leq \sum_{j\in B,\, j\leq i} \frac{\hat{x}_{j,t}}{p_t\cdot (1-\hat{s}_{j,t})}\right\}$
		        \EndIf
		    \EndFor
		    \State $P_t\gets C_t\setminus V(\calM)$ \Comment{Bidders =  free candidates}
		    \If{$P_t\neq \emptyset$ \textbf{ and } $t$ arrived}
		    \State pick some $i\in \arg\max\{w_{i,t} \mid i\in P_t\}$
            \State $\calM\gets \calM\cup \{(i,t)\}$
		    \EndIf
        \EndFor
        \State \textbf{Output} $\calM$
	\end{algorithmic}
\end{algorithm}

\begin{fact}
\Cref{alg:stochastic} is well-defined. In particular, $\frac{\hat{x}_{i,t}}{p_t\cdot (1-\hat{s}_{i,t})} \leq 1$ for every pair $(i,t)$.
\end{fact}
\begin{proof}
By \Cref{hat-degrees}, we have that $\hat{s}_{i,t} = s_{i,t}\cdot (1-\eps) + (s_{i,t}-\theta)^+\cdot (\eps+\delta) \leq s_{i,t}$. Therefore, by our choice of $\theta=\frac{\delta}{\eps+\delta}$ and LP Constraint \eqref{OPTon-constraint}, we have that, indeed
\begin{align*}
\frac{\hat{x}_{i,t}}{p_t\cdot (1-\hat{s}_{i,t})} & \leq  \begin{cases} 
\frac{{x}_{i,t}\cdot (1-\eps)}{p_t\cdot (1-{s}_{i,t})} \leq \frac{{x}_{i,t}}{p_t\cdot (1-{s}_{i,t})} \leq 1 & s_{i,t}\leq \theta
\\ 
\frac{{x}_{i,t}\cdot (1+\delta)}{p_t\cdot (1-{s}_{i,t}(1+\delta) + \delta)} = \frac{{x}_{i,t}}{p_t\cdot (1-{s}_{i,t})} \leq 1 & s_{i,t}>\theta.
\end{cases} \qedhere 
\end{align*}
\end{proof}

\Cref{alg:stochastic} is clearly a polytime algorithm. It remains to analyze its approximation ratio.

The following lemma makes explicit a fact implicit in the proof of the main theorem of \cite{braverman2022max}: specifically, that a per-online-node ``approximation ratio'', relating  $w(\calM(t),t)$---the weight obtained by matching $t$---to the LP solution's value from $t$, implies a global approximation ratio of the same value. We provide a short self-contained proof of this fact for completeness.
\begin{lem}\label{per-online-to-global}
	If for each online node $t$ and $z\geq 0$ we have that $\Pr[w(\calM(t),t)\geq z]\geq \sum_{i: w_{i,t}\geq z} \alpha \cdot x_{i,t}$, then the algorithm is an $\alpha$-approximation of the optimum online algorithm.
\end{lem}
\begin{proof}
Denote by $w_t:=w(\calM(t),t)$ the weight of the matched edge of $t$, with $w_t = 0$ if $t$ is not matched (possibly due to non-arrival). Then, the weight of the output matching $\calM$ satisfies
\begin{align*}
\sum_t \E[w_t] & = \sum_t \int_{z=0}^{\infty} \Pr[w_t\geq z]\, dz \geq \sum_t \int_{z=0}^{\infty} \alpha\cdot \sum_{i: w_{i,t}\geq z} x_{i,t} \, dz = \alpha\cdot \sum_{i,t} w_{i,t}\cdot x_{i,t}.
\end{align*}
The proof is then completed by  \Cref{lem:LP-bound}, implying that the RHS upper bounds $\alpha$ times the value of the optimum online algorithm's gain.
\end{proof}

The objective of our analysis will therefore be to lower bound $\Pr[w(\calM(t))\geq z]$ for all $z\geq 0$.\\

In our analysis, we let $F_t=[n]\setminus V(\calM)$ be the set of free offline nodes at time $t$, and let
$F_{i,t} = \mathds{1}[i\in F_t] = 1-\sum_{t'<t}\mathds{1}[(i,t')\in \calM]$ indicate whether $i$ is free by time $t$. 
Similarly, we let $E_{i,t}=1-F_{i,t}$ indicate whether $i$ was matched (``\emph{engaged}'') before time $t$.

Our first observation is that our modification to the notion of freedom results in the following lower bound on the probability of a node being free at time $t$.
\begin{obs}\label{lem:basic-bucketing-props-stochastic} 
For each time $t$ and offline $i\in[n]$, we have that
$\Pr[F_{i,t}]\geq 1-\hat{s}_{i,t}$.
\end{obs}
\begin{proof}
Let $B^1_{i,t}$ indicate that $i$ bid for the first time before time $t$. 
We prove that $\E[B^1_{i,t}]=\hat{s}_{i,t}$, by induction on $t\geq 1$. The base case, is trivial, as $F_{i,1}\equiv 1$ and $\hat{s}_{i,1}\equiv 0$.
Now, by induction, we have 
$$\E[B^1_{i,t+1}-B^1_{i,t}] = p_t\cdot \frac{\hat{x}_{i,t}}{p_t\cdot(1-\hat{s}_{i,t})}\cdot (1-\E[B^1_{i,t}]) = \hat{x}_{i,t},$$
where the first equality relies on independence of $t$'s arrival, $i$ becoming a candidate at time $t$, and $i$ not having yet bid, and the second equality follows from the inductive hypothesis.
By linearity, we then have that $\E[B^1_{i,t+1}]=\hat{s}_{i,t+1}$.
Since $1-F_{i,t} = E_{i,t}\leq B^1_{i,t}$, the claim follows.
\end{proof}

Next, for any time $t$ and offline node set $S\subseteq [n]$, we let $F_{S,t} = \bigwedge_{i\in S} F_{i,t}$ and $E_{S,t} = \bigwedge_{i\in S} E_{i,t}$ 
indicate whether all of $S$ is free (respectively, engaged) by time $t$.
So far, we have established that $\Pr[E_{i,t}]\leq \hat{s}_{i,t}$. The following lemma, whose statement and proof mirror that of \cite[Lemma 3.3]{braverman2022max} (who use $x_{i,t}$ and $s_{i,t}$, as opposed to $\hat{x}_{i,t}$ and $\hat{s}_{i,t}$), asserts that this upper bound is ``sub-multiplicative''.

\begin{lem}\label{stochastic-cylinder-bound}
	For each time $t$ and set of offline nodes $S\subseteq[n]$, we have that
	$$\Pr[E_{S,t}]\leq \prod_{i\in S} \hat{s}_{i,t}.$$
\end{lem}
To prove the above, we need the following claim from  \cite{braverman2022max}. For completeness, we provide a short probabilistic proof of this lemma (see  \cite[Lemma 3.2]{braverman2022max} for a longer algebraic proof).

\begin{lem}\label{BDM-algebraic-lemma}
    For all $i\in [n]$, let $q_i\in [0,1]$. Then,
    $$\sum_{A\subseteq S} \Pr[E_{A,t},F_{S\setminus A,t}]\cdot \prod_{i\in S\setminus A} q_i = \sum_{A\subseteq S} \Pr[E_{A,t}]\cdot \prod_{i\in A} (1-q_i) \prod_{i\in S\setminus A} q_i.$$
\end{lem}
\begin{proof}
    Consider a probability space where at time $t$ each offline node $i$ tosses a coin with probability of heads $q_i$, independently.
    The LHS corresponds to all free nodes having their coin come up heads.
    The RHS corresponds to the same event, i.e., all nodes whose coin came up tails were engaged. 
\end{proof}

\begin{proof}[Proof of \Cref{stochastic-cylinder-bound}]
	The proof is by induction on $t\geq 1$ for all sets $S$. The base case $t=1$, for which $s_{it}\equiv 0$ for all $i$, is trivial. For the inductive step, 
	denote the marginal probability of $i$ bidding at time $t$ (requiring $t$'s arrival) by $$m_{i,t} := \Pr[i\in P_t] = \frac{\hat{x}_{i,t}}{1-\hat{s}_{i,t}}.$$
	We note that with this terminology, we have that
    \begin{align}\label{U_it-recurrence}
	    \hat{s}_{i,t+1} & = \hat{s}_{i,t} + \hat{x}_{i,t} = \hat{s}_{i,t} + m_{i,t}\cdot (1-\hat{s}_{i,t}) = \hat{s}_{i,t}\cdot (1-m_{it}) + m_{it}. 
	\end{align}
	The inductive step follows by noting that at most one of the nodes $i\in S$ may be matched at time $t$, and this requires that the single free node $i\in S$ bid for $t$, which happens with probability $m_{i,t}$ (conditioned on any event determined by randomness  up to time $t$ that implies that the node is free), as follows.
	\begin{align*}
		\Pr[E_{S,t+1}] & \leq \sum_{\substack{A\subseteq S \\ |S\setminus A|\leq 1}} \Pr[E_{A,t}, F_{S\setminus A,t}] \cdot \prod_{i\in S\setminus A} m_{i,t}\\
		& = \sum_{\substack{A\subseteq S \\ |S\setminus A|\leq 1}} \Pr[E_{A,t}] \cdot \prod_{i\in A} \left(1-m_{i,t}\right) \prod_{i\in S\setminus A} m_{i,t} & \textrm{\Cref{BDM-algebraic-lemma}}\\
		& \leq \sum_{A\subseteq S} \Pr[E_{A,t}] \cdot \prod_{i\in A} \left(1-m_{i,t}\right) \prod_{i\in S\setminus A} m_{i,t} \\
		& \leq \sum_{A\subseteq S} \prod_{i\in A} U_{i,t} \cdot \prod_{i\in A} \left(1-m_{i,t}\right) \prod_{i\in S\setminus A} m_{i,t} & \textrm{I.H.} \\
		& = \prod_{i\in S}\left(\hat{s}_{i,t} \cdot (1-m_{it}) + m_{it}\right) & \textrm{binomial expansion} \\
		& \leq \prod_{i\in S} \hat{s}_{i,t+1} & \textrm{Equation \eqref{U_it-recurrence}}. &  \qedhere
	\end{align*}
\end{proof}

The preceding lemma allows us to prove the following counterpart to \Cref{proposal-prob-perturbed-rephrased}, key to our improved approximation ratios.

\begin{lem}\label{proposal-prob-perturbed-rephrased-stochastic}
	For any time $t$ and set of offline nodes $S\subseteq[n]$, \Cref{alg:stochastic} satisfies
	$$\Pr[S\cap P_t \neq \emptyset] \geq 1 - \prod_{B\in \calB_t}\left(1-\sum_{i\in B\cap S} \hat{x}_{it}/p_t\right).$$
\end{lem}
\begin{proof}
    Let $C_t$ and $P_t$ be as in \Cref{alg:stochastic}.
    Since $P_t = C_t \setminus \{i\mid E_{i,t}=1\}$, there are no bids from set $S$ ($S\cap P_t=\emptyset$) iff all $i\in A=C_t\cap S$ are no longer free, i.e., iff $E_{A,t}$ holds.
    So, by total probability over $C_t\cap S$, we have that  $(\star):=\Pr[t \textrm{ arrives and } S\cap P_t = \emptyset] $ satisfies
    \begin{align*}
       (\star) &  = \sum_{A\subseteq S} \Pr[C_t\cap S = A] \cdot \Pr[E_{A,t}]\\
        & = \sum_{A\subseteq S} \prod_{B: |B\cap A|=0} \Pr[C_t \cap B \cap A = \emptyset] \prod_{B: |B\cap A|=1}\sum_{i\in B\cap A} \Pr[C_t\cap B\cap A = \{i\}] \cdot \Pr[E_{A,t}] \\
        & \leq \sum_{A\subseteq S} \prod_{B: |B\cap A|=0} \Pr[C_t \cap B \cap A = \emptyset] \prod_{B: |B\cap A|=1}\sum_{i\in B\cap A} \Pr[C_t\cap B\cap A = \{i\}] \cdot \Pr[E_{i,t}] \\
        & = \prod_{B\in \calB_t} \left(\Pr[C_t\cap B\cap S = \emptyset] + \sum_{i\in B\cap S} \Pr[C_t\cap B=\{i\}]\cdot \Pr[E_{i,t}]\right) \\
        & = \prod_{B\in \calB_t} \left(1-\sum_{i\in B\cap S} \Pr[C_t\cap B=\{i\}]\cdot \Pr[F_{i,t}]\right) \\
        & \leq \prod_{B\in \calB} \left(1-\sum_{i\in B\cap S} \hat{x}_{i,t}/p_t \right).
    \end{align*}
    Above, the second and third equality follow from independence of the sets $\{C_t\cap B \mid B\in \calB_t\}$ and from $|C_t\cap B|\leq 1$ always. The first inequality follows from \Cref{stochastic-cylinder-bound}. The fourth equality follows  from $|C_t\cap B|\leq 1$, and $F_{i,t} = 1-E_{i,t}$.
    Finally, the last inequality follows from $\Pr[C_t\cap B=\{i\}]=\frac{\hat{x}_{i,t}}{p_t\cdot (1-\hat{s}_{i,t})}$ and \Cref{lem:basic-bucketing-props-stochastic}, implying that $\Pr[C_t\cap B=\{i\}]\cdot \Pr[F_{i,t}]\geq \frac{\hat{x}_{i,t}}{p_t\cdot (1-\hat{s}_{i,t})}\cdot (1-\hat{s}_{i,t}) \geq \hat{x}_{i,t}/p_t$.
\end{proof}

The following lemma, which is a counterpart to \Cref{conditional-ratio} for the non-stochastic setting, allows us to obtain a concrete lower bound from \Cref{proposal-prob-perturbed-rephrased-stochastic}.

\begin{lem}\label{conditional-ratio-stochastic}
	If $\epsilon,\delta \in [0,1]$ guarantee that for all $z\geq \frac{1-\theta}{2} = \frac{\epsilon}{2(\epsilon+\delta)}$,
	\begin{align}\label{eqn:conditional-ratio-stochastic}
		f_{\epsilon,\delta}(z) := \exp(-z\cdot (1+\delta)) - (1-z\cdot (1-\eps)) \geq 0,
	\end{align}
	then \Cref{alg:stochastic} with parameters $\epsilon$ and $\delta$ satisfies for each time $t$ and set $S\subseteq [n]$
	$$\Pr[S\cap P_t\neq \emptyset]\geq \left(1-\exp\left(-1-\delta+\frac{\eps+\delta}{1-\eps}\right)\cdot\frac{1-\eps}{1+\delta}\right)\cdot \sum_{i\in S} x_{i,t}/p_t.$$
\end{lem}
The proof is near identical to that of \Cref{conditional-ratio}, and is therefore omitted.\footnote{The only changes to the proof are (1) the syntactic generalization of \Cref{first-fit} implies that all but one bin $B\in \calB_t$ computed in \Cref{line:FF-stochastic} has $\sum_{i} \hat{x}_{i,t}/\mathbf{p_t}\geq \frac{1-\theta}{2}$, and (2) the convexity argument requires that we add dummy nodes $J$ with $s_{j,t}=\theta$ for all $j\in J$ and $\sum_{j\in J} x_{j,t} = p_t - \sum_{i} x_{i,t}$, where this sum is non-negative, by LP Constraint \ref{LP:flow-for-t}.}

Finally, the same optimization over $\eps,\delta$ as in \Cref{sec:rounding-matchings} yields our main result of this section.

\begin{thm}
	\Cref{alg:stochastic} with parameters $(\eps,\delta) \approx (0.0480,0.0643)$ is a polytime online stochastic weighted bipartite matching algorithm that $\ratio$-approximates the optimum online algorithm.
\end{thm}
\begin{proof}
    Let $\alpha=1-\exp\left(-1-\delta+\frac{\eps+\delta}{1-\eps}\right)\cdot\frac{1-\eps}{1+\delta} \geq \ratio$. 
    By our algorithm's greedy matching choice, $t$ gets matched to a node $i$ with $w_{i,t}\geq w$ iff $t$ arrives and gets a bid from the set $S_{t,w}:=\{i \mid w_{i,t}\geq w\}$. So, by \Cref{conditional-ratio-stochastic}, as $(\eps,\delta)$ above satisfy \Cref{eqn:conditional-ratio-stochastic} for all $z\geq 0$ (see the proof of \Cref{thm:ODRS}), 
    $$\Pr[w(\calM(t),t)\geq w] = \Pr[t \textrm{ arrives} \wedge P_t\cap S_{t,w} \neq \emptyset] \geq p_t\cdot \alpha\cdot \sum_{i\mid w_{i,t}\geq w}x_{i,t}/p_t = \alpha\cdot \sum_{i\mid w_{i,t}\geq w} x_{i,t}.$$
    The approximation ratio of \Cref{alg:stochastic} then follows from \Cref{per-online-to-global}, while the algorithm's polynomial running time is immediate.
\end{proof}
\color{black}

\subsubsection{A simple direct application: Algorithms with predictions}\label{sec:predictions}
Our ODRS of \Cref{thm:arbitrary} finds a direct application in the burgeoning area of \emph{online algorithms with advice} (see the survey \cite{MitzVas20} and site \cite{site:ALPS}). Here, an online algorithm is equipped with machine-learned (ML) predictions concerning the input.
For example, many works show that for related problems (including bipartite matching), high-quality fractional solutions are efficiently learnable  
and can guide integral algorithms, provided an effective online rounding scheme is given \cite{lavastida2021learnable,lattanzi2020online,li2021online}.
We add to this line of work, as follows: By linearity of expectation, our ODRS can round any fractional \emph{edge-weighted} bipartite matching while preserving its value up to a ratio of $\ratio$.
This immediately yields the first online \emph{edge-weighted} bipartite matching algorithm with predictions, getting a competitive ratio of $(1-1/e+\Omega(1))$
with sufficiently good predictions.
In contrast, with no predictions, only recently was it shown how to break the barrier of $1/2$ for this problem (using \emph{free disposal}) \cite{fahrbach2020edge,blanc2021multiway,gao2021improved}, and $1-1/e+\Omega(1)$ is impossible to achieve even in unweighted graphs \cite{karp1990optimal}.
\color{black}

\subsection{Application of the level-set ODRS}
We now turn to discussing applications of our main level-set algorithm, and its strong negative correlation properties.

\subsubsection{Multi-stage stochastic optimization} 
\label{sec:app-stoch-opt} 
We consider an application to multi-stage stochastic optimization, building on the framework of \cite{DBLP:conf/soda/Srinivasan07}. We recall briefly that in such multi-stage problems, various parameters of an optimization problem materialize stochastically over some number $k$ of stages, where each stage contains stochastic information about the following stage(s) and wherein we can take actions---with actions cheaper in earlier stages, but possibly inaccurate since we only know the future stochastically. 

Motivated by the results of \cite{DBLP:journals/siamcomp/SwamyS12}, a broad class of multi-stage stochastic covering problems can be cast as the following family of \emph{online} problems: see Section 2 of  \cite{DBLP:conf/soda/Srinivasan07}, from which we quote the framework almost verbatim next. There is a hidden covering
problem ``minimize $c^T \cdot x$ subject to $Ax \geq b$ and with all
variables in $x$ being non-negative integers'' that is revealed online as follows.
We let $m$ denote 
the number of rows of $A$; also, the variables in $x$
are indexed as $x_{j,\ell}$, where $1 \leq j \leq n$ and $1 \leq \ell \leq k$. (We interchange $m$ and $n$ from 
\cite{DBLP:conf/soda/Srinivasan07}.)
Such a covering problem, as well as a feasible fractional solution $x^*$ for it (where each $x^*_{j,\ell}$ is allowed to be a non-negative \emph{real}), are revealed to us in $k$ stages as follows. In each stage $\ell$ where $1 \leq \ell \leq k$, we are given the $\ell^{th}$-stage fractional values $(x^*_{j,\ell}:~1 \leq j \leq n)$ of the variables along with their respective columns in the coefficient matrix $A$, and their respective coefficients in the objective-function vector $c$. 
Given this setup, we need to round the variables 
$(x^*_{j,\ell}:~1 \leq j \leq n)$ immediately and irrevocably at stage $\ell$, using randomization if necessary. (Note the direct 
connection to online computation.) The goal is to develop such a rounded vector $(s_{j,\ell}:~1 \leq j \leq n, 1 \leq \ell \leq k)$ that satisfies the 
constraints $Ay \geq b$, and whose (expected) approximation ratio 
$\E[c^T \cdot y] / c^T \cdot x^*$ is small, where the expectation is only over any random choices made by our online algorithm in this arrival model. This completes the
description of the framework from \cite{DBLP:conf/soda/Srinivasan07}.

We make two contributions in this framework. A motivating family of problems to consider is vertex-(multi-)cover on graphs and hypergraphs. The multi-stage stochastic version of the classical vertex-cover problem on a given graph $G = (V, E)$ is basically as follows: the coverage constraint corresponding to an edge $e = (u,v) \in E$ is now
\begin{equation}
\label{eqn:stoch-vc} 
    \left(\sum_{\ell = 1}^k x_{u,\ell}\right) + 
\left(\sum_{\ell = 1}^k x_{v,\ell}\right)
\geq 1.
\end{equation}
 Improving on the $2k$-approximation of 
\citet{DBLP:journals/siamcomp/SwamyS12} for this problem---in the ``online rounding'' model mentioned in the previous paragraph---a $2$-approximation is developed in \cite{DBLP:conf/soda/Srinivasan07}. 

Our \emph{first contribution} is the following generalization of the $2$-approximation for multi-stage stochastic vertex cover from  \cite{DBLP:conf/soda/Srinivasan07}. Consider the following more-general problem in $d$-regular hypergraphs $H = (V,E)$ (or hypergraphs with all edge-sizes lower bounded by $d$): for some integer parameter $t$, we want the vertices in each edge $e \in E$ to be covered in total to an extent of at least $t$, where vertices can be multi-covered. That is, generalizing \eqref{eqn:stoch-vc}, the coverage constraint corresponding to an edge $e = \{v_1, v_2, \ldots, v_d\} \in E$ is 
\begin{equation}
\label{eqn:stoch-hyp-multi-vc} 
 \sum_{i=1}^d  \sum_{\ell = 1}^k x_{v_i,\ell} 
\geq t,  
\end{equation}
where each $x_{v_i,\ell}$ is allowed to be any non-negative integer; note that the above-seen stochastic vertex cover is the special case where $d = 2$ and $t=1$. 
Recalling our online arrival model, we need to round the LP solution  $(x^*_{v,\ell}:~v \in V)$ right away at stage $\ell$, for each $\ell \in [k]$. We proceed as follows to obtain an $\alpha := (d + t-1) / t$--approximation in this model, generalizing the $2$-approximation for $d = 2$ and $t=1$. In particular, if $t$ is large, this is essentially a $1$-approximation.

Our algorithm works as follows.
For each vertex $v \in V$, consider the scaled vector $z(v) := (\alpha \cdot x^*_{v,\ell}:~\ell \in [k])$ and independently for all $v$, run \Cref{alg:rounding} on $z(v)$ to obtain the final rounded values 
$X_{v,\ell}:~\ell \in [k])$. (As usual, we imagine padding this vector with a dummy final element so that the sum of the entries is an integer.) We first verify that \eqref{eqn:stoch-hyp-multi-vc} is satisfied with probability one for each $e = \{v_1, v_2, \ldots, v_d\} \in E$: 
\begin{align}
\sum_{i=1}^d  \sum_{\ell = 1}^k X_{v_i,\ell} & \geq
\sum_{i=1}^d  \left\lfloor \sum_{\ell = 1}^k \alpha \cdot x^*_{v_i,\ell} \right\rfloor & \mbox{(by Property \ref{prop:floor/ceil})} 
\label{eqn:hypvc1} \\
& > 
\sum_{i=1}^d  \left( \left(\sum_{\ell = 1}^k \alpha \cdot x^*_{v_i,\ell}\right)  - 1 \right) & \mbox{($\lfloor \beta \rfloor > \beta - 1$ for all $\beta$)} \label{eqn:hypvc2} \\
& \geq \alpha \cdot t - d & \mbox{(since $x^*$ satisfies \eqref{eqn:stoch-hyp-multi-vc})} \nonumber \\
& = t-1, \nonumber 
\end{align} 
which implies that 
$\sum_{i=1}^d  \sum_{\ell = 1}^k X_{v_i,\ell} \geq t$, as required, since the LHS here is an integer and because of the strict inequality connecting \eqref{eqn:hypvc1} and \eqref{eqn:hypvc2}. We thus obtain an online solution to the rounding problem which satisfies all constraints with probability one, and whose expected objective function value is at most $\alpha (c^T \cdot x^*)$. We observe that the work of \cite{DBLP:conf/soda/Srinivasan07} also runs an online rounding algorithm, but that the problem considered there is much simpler since all we require there (in place of Property~\ref{prop:floor/ceil}) is that if the entries of the given input vector sum to at least one, then at least one entry is rounded to one with probability one. In particular, this does not address the multi-coverage constraint we satisfy here.

Our \emph{second contribution} is that because of Property~\ref{prop:sr}, the objective function is strongly concentrated around its mean (given by the Chernoff bound). This is an especially useful property in areas like stochastic optimization, where \emph{the algorithm can be run only once}. This is in comparison to \emph{offline} minimization problems with a non-negative objective where we can repeat the algorithm $O(1/\epsilon)$ times and choose the best solution obtained, and apply Markov's inequality to each iteration to show that the best solution is at most $(1 + \epsilon)$ times the expected value with high probability. Thus, the fact that \Cref{alg:rounding} yields sharp concentration as implied by Property~\ref{prop:sr}---and does not just guarantee Properties \ref{prop:marginals} and \ref{prop:floor/ceil}---is very useful in the stochastic-optimization context.

Summarizing the preceding discussion (and formalizing the preceding paragraph), we obtain the following result.
\begin{thm}
For any $k\geq 1$, there exists a $k$-stage stochastic hypergraph vertex cover algorithm for $d$-regular hypergraphs where each hyper-edge must be covered $t$ times, with approximation ratio that is  $\alpha=\frac{d+t-1}{t}$ in expectation and is $\alpha(1+o(1))$ w.h.p.~if costs are polynomially bounded and $OPT= \omega(\log n)$.
\end{thm}
\begin{proof}
Our algorithm satisfies the (multi-)coverage constraint, by the above, and has expected approximation ratio $\alpha$ by construction. 
The high-probability bound follows from standard Chernoff bounds for negatively-associated variables, relying on Property \ref{prop:sr} and closure of NA under products (\Cref{na-closure}), relevant since we run independent copies of \Cref{alg:rounding} for each vertex.
\end{proof}

\subsubsection{Negative association, fairness, diversity, and streaming}
\label{sec:app-inter-fair} 
While Section~\ref{sec:app-stoch-opt} only required the fact that Property \ref{prop:sr} yields sharp tail bounds for sums of the $X_i$ with non-negative weights, we go further now and use the negative association guaranteed by Property \ref{prop:sr}, which helps us correlate multiple such sums as well as handle submodular objectives well. The submodular-objective application further benefits from the fact that \Cref{alg:rounding} can also be viewed as a \emph{data-stream} algorithm: note that at when $x_t$ arrives, the algorithm need only remember $s_{t-1}$ and $S_t$, and hence can be implemented with $O(\log n)$ space. 

\emph{Intersectionality} is a key notion in the study of fairness, where the intersection of multiple attributes (e.g., race, gender, age) can impact the outcomes of a person in a more fine-grained manner than does standard group-fairness (where we typically take a single attribute). Its study has naturally impacted AI/ML fairness as well---see, e.g., \cite{pmlr-v81-buolamwini18a} and its followups---and is a topic of much debate in terms of precise definitions (see, e.g., \cite{kong:intersectional-fairness}). We make a contribution to this nascent area, while being aware that many such formulations/contributions in the still-early stages of research in AI/ML fairness are speculative. 

Consider people arriving online and requesting a resource whose availability expands over time: e.g., vaccines for a new pathogen or a popular new model of car. It is anticipated that by time $t$, at most $a_t$ units of this resource will be available. An ML system that adapts to the changing realities on the ground (e.g., how the pathogen is spreading and which communities seem to be impacted the most) decides online, the probability $x_t$ with which to allocate the resource to the request arriving at time $t$; we naturally require $\sum_{i=1}^t x_i \leq a_t$ for all $t$. We respond to this in real time by constructing $X_t \in \{0,1\}$ using \Cref{alg:rounding}; Property \ref{prop:floor/ceil} ensures that we never exceed the supply-limits $a_t$. (In reality, a \emph{batch} of requests will arrive at time $t$, which of course \Cref{alg:rounding} naturally extends to.) The ML system outputting the vector $x$ may be (required to be) unware of protected attributes that characterize our underlying demographic groups and their intersections: given $x$, can we show that while intersectional unfairness may be present in $x$ (in which case the ML system needs to be refined), our rounding algorithm tries to at least \emph{limit the scope} of such intersectional unfairness in some way?

We formulate this problem as follows. Suppose each person requesting the resource has three attributes $A_1, A_2$, and $A_3$: the ``three" here is just for simplicity, and it is easy to see that the framework below generalizes to any number of attributes. Suppose $U_1, U_2, \ldots, U_{\alpha}$ is some given partition of the population according to their values for $A_1$ (e.g., if $A_1$ is age, the $U_i$'s may be disjoint age-intervals). Similarly, suppose we have a partition of the population according to $A_2$ as $V_1, V_2, \ldots, V_{\beta}$ and according to $A_3$ as $W_1, W_2, \ldots, W_{\gamma}$. In order to model intersectionality, let us define the subgroup $G_{i,j,k} := U_i \cap V_j \cap W_k$. Let $s_{i,j,k} = \sum_{r \in G_{i,j,k}} X_r$ be the random variable denoting the amount of resource received in total by members of $G_{i,j,k}$. We ask: ``While intersectional unfairness may be present, say inherently in the vector $x$, does our rounding at least limit the extent of it?" That this is indeed true follows from Property \ref{prop:sr} which guarantees negative association, as follows. For any sequence of thresholds $(t_{i,j,k})$, 
the probabilities of different subgroups $G_{i,j,k}$ receiving intersectional unfairness in the sense of $s_{i,j,k} \leq t_{i,j,k}$, can be bounded as well as if the $X_r$'s had been independent: 
\begin{itemize}
    \item $\forall S \subseteq \left([\alpha] \times [\beta] \times [\gamma]\right), ~\Pr[\bigwedge_{(i,j,k) \in S} (s_{i,j,k} \leq t_{i,j,k})] \leq \prod_{(i,j,k) \in S} \Pr[s_{i,j,k} \leq t_{i,j,k}]$; and, more generally, 
    \item in disease-spread and meme-spread-like contexts, we are often interested in monotone non-linear functions that model phase transitions: e.g., if a heavily-interacting subgroup receives fewer than a threshold of vaccines, an explosive epidemic could happen within that subgroup---with monotone behavior away from this threshold. Given any sequence of monotone functions (all increasing or all decreasing) $(f_{i,j,k})$, we have
    \[ \forall S \subseteq \left([\alpha] \times [\beta] \times [\gamma]\right), ~\E\left[\prod_{(i,j,k) \in S} f_{i,j,k}(s_{i,j,k})\right] \leq \prod_{(i,j,k) \in S} \E\left[f_{i,j,k}(s_{i,j,k})\right]. \] 
    As is well-known, letting $\chi(\cdot)$ be the indicator function,
    this implies that $\sum_{(i,j,k) \in S} \chi(s_{i,j,k} \leq t_{i,j,k})$ has Chernoff-like concentration around its mean; in particular, while some subgroups in $S$ may face intersectional unfairness (that needs to be separately analyzed by inspecting $x$), it is unlikely that many do. 
\end{itemize}

We next present an application to diversity in search results. In addition to the classical interest in this problem in information retrieval \cite{boyce-IR}, there has been much recent work in 
search-result diversification, motivated, e.g., by the fact that the document-set retrieved should account for the interests of the
user population \cite{DBLP:journals/sigir/CarbinellG17,DBLP:conf/sigir/ClarkeKCVABM08}. By developing a model of knowledge about the topics the query or the documents may refer to, the work of 
\cite{DBLP:conf/wsdm/AgrawalGHI09} presents an $(1 - 1/e)$--approximation for their objective of maximizing average-user satisfaction with the search results. Given a search query $q$, a corpus of $N$ documents $\mathcal{D}$, and a bound $k$ on the number of documents to retrieve from $\mathcal{D}$ for $q$, a monotone submodular function $f: 2^{\mathcal{D}} \rightarrow [0,1]$ is developed in \cite{DBLP:conf/wsdm/AgrawalGHI09}, where 
$f(T)$ is the (approximate) mix of relevance and diversity if the set $T$ with $|T| \leq k$ is then returned to the user. The problem is thus to approximately maximize $f(T)$ subject to $|T| \leq k$. Consider the setting where $\mathcal{D}$ is very large, we only get streaming access to it, and where, after encountering document $d \in \mathcal{D}$, we get ML advice on the probability $x_d$ with which to choose $d$. \Cref{alg:rounding} is tailor-made for this, as it needs only $O(\log N)$ space. Furthermore, by the already-known \Cref{thm:stochastic-order}, we obtain that our final expected value $\E[f(X_1, X_2, \ldots, X_N)]$ is at least as large as if we had instead chosen the $X_i$'s independently with marginals $x_i$. 
Moreover, by \Cref{submod-concentration}, the  subdmodular objective value obtained exhibits strong lower tail bounds (i.e., is unlikely to fall much below its expectation).

	\section{Lower Bounds} \label{sec:lower-bound}
In this section we present some impossibility results for ODRSes.

A simple example in \cite{devanur2013randomized} rules out a rounding ratio of one, or even just greater than $\frac{7}{8} = 0.875$. In this section we improve this bound to $(2\sqrt{2}-2) \approx 0.828$ with a similarly simple example. 

Before providing our lower bound, we provide an even simpler example ruling out a rounding ratio of one. This example highlights the need of ODRSes with high rounding ratio to create some form of negative correlation between all subsets of offline nodes, motivating our lower bound example. (We are also hopeful that this insight may inform future ODRSes.)

\paragraph{Example.} 
Suppose an ODRS $\calA$ with rounding ratio of one exists, and apply it to the following input on $n=3$ offline nodes and 4 online nodes: For $k=1,2,3$, online node $k$ only neighbors offline node $k$, and $x_{kk}=1/2$. So, each offline node's matched status so far is a $\Ber(1/2)$ random variable. 
Finally, online node $t=4$ arrives, neighboring some two offline nodes, $i,j$, and $x_{it}=x_{jt}=1/2$. The above is clearly a feasible fractional matching, but since we assumed that $\calA$ has a rounding ratio of one, it must match $t$ with probability one. Since $t$ cannot be matched if both $i$ are $j$ previously matched, and both are matched with probability $1/2$ before time $t$, this requires $i$'s and $j$'s matched status before time $t$ to be perfectly negatively correlated. That is, $i$ is unmatched iff $j$ is matched, and vice versa. Since $i,j$ could be any two offline nodes, this requires perfect negative correlation between any two of these three variables, which is impossible.\footnote{Perfect negative pairwise correlation between three binary variables $X_1,X_2,X_3\in \{0,1\}$ can be stated succinctly as $X_1+X_2 = X_1+X_3 = X_2+X_3=1$, but this linear system only has a \emph{fractional} solution $X_1=X_2=X_3=1/2$.} Thus, we reach a contradiction, and an ODRS with rounding ratio of one is impossible.

To generalize the above and obtain a concrete numeric bound, we need the following fact, whereby (near-)positive pairwise correlation is unavoidable in a large set of Bernoulli random variables. (See \Cref{almost-positive-cylinder-dependence} for a potentially useful generalization to $k$-wise correlations.)
\begin{fact}\label{pairwise-positive-correlation}
	Let $Y_i\sim Ber(p)$ for $i=1,\dots,n$ be (possibly dependent) Bernoulli variables. Then, $$\max_{i,j}\Cov(Y_i,Y_j)\geq -2p/(n-1).$$
\end{fact}
\begin{proof}
	This bound follows from non-negativity of variance and the following, after rearranging.
	\begin{align*}
		0 & \leq \Var\left(\sum_i Y_i\right) = \sum_{i} \Var(Y_i) + \sum_{i,j}\Cov(Y_i,Y_j) \leq n\cdot p + 
		{n \choose 2}\max_{ij}\Cov(Y_i,Y_j). \qedhere
	\end{align*}
\end{proof}

\begin{lem}\label{lem:LB}
	Any online rounding algorithm for bipartite fractional matchings has some fractional matching $\mathbf{x}=\frac{1}{2}\cdot \mathbf{1}^E$ on graph $G=(V,E)$ such that for some edge $e\in E$, the output matching $\calM$ satisfies
	$$\Pr[e\in \calM] \leq (2\sqrt{2}-2)\cdot x_e  \approx 0.828 \cdot x_e.$$
\end{lem}
\begin{proof}
	Consider an input with $x_e=1/2$ for each edge, with the first $n$ online nodes neighboring two distinct offline nodes, with one last online node neighboring two (judiciously chosen) offline nodes, from two prior online nodes' neighborhoods.
	We wish to show that for some instantiation of this family of $\mathbf{x}$ and $G$, some online node (and hence some edge) is matched with low probability. Specifically, letting $p=2\sqrt{2}-2$, we wish to show that for some such $\mathbf{x}$ and $G$, 
	\begin{align}\label{eqn:unmatched-node}
		\exists t\in [n+1] \textrm{ such that } \Pr[t\in V(\calM)]\leq p+\frac{p}{2(n-1)},
	\end{align}	
	Since each online node $t$ has $\sum_i x_{it}=1$, this implies that some edge $e$ is matched with probability at most $(2\sqrt{2}-2+\frac{p}{2(n-1)})\cdot x_e$, by linearity of expectation.
	Taking $n$ to be sufficiently large rules out any $(2\sqrt{2}-2+\epsilon)$-approximate rounding. We turn to prove \Cref{eqn:unmatched-node}.
	
	If any of the first online nodes is matched with probability less than $p$, we are done. Otherwise, let $Y_t = \mathds{1}[t\in V(\calM)]$ be indicators for online node $t$ being matched. We couple these with variables $Z_t \geq \Ber(p)$ such that $Y_t\geq Z_t$ always. By \Cref{pairwise-positive-correlation}, there exist two online nodes $t\neq t'$ with $\Cov(Z_t,Z_{t'})\geq -2p/(n-1)$.
	Let $N(t)=\{i_1,i_2\}$ and $N(t')=\{i_3,i_4\}$ be the neighborhoods of $t$ and $t'$, respectively. By a simple averaging argument, some pair of nodes $(i,j)\in N(t)\times N(t')$ are both matched before the online last arrival with probability at least 
	$$\frac{\Pr[Y_t,Y_{t'}]}{4} \geq \frac{\Pr[Z_t,Z_{t'}]}{4} \geq  \frac{\Pr[Z_t]\cdot \Pr[Z_{t'}]-2p/(n-1)}{4} \geq \frac{p^2 -2p/(n-1)}{4}.$$ 
	So, if the last online node neighbors $\{i,j\}$, it is matched with probability at most $1-\frac{p^2 -2p/(n-1)}{4}$.
	However, as $p=2\sqrt{2}-2$ is a root of $1-y-\frac{y^2}{4}$, this probability is precisely $1-\frac{p^2}{4} + \frac{p}{2(n-1)} = p + \frac{p}{2(n-1)}$, as desired.
\end{proof}

\paragraph{Inevitable near-positive cylinder dependence.}
A useful extension of \Cref{pairwise-positive-correlation} of possible independent interest argues that for any (sufficiently large) subset of Bernoulli variables, some subset of these variables must be (nearly) positively cylinder dependent. In \Cref{appendix:lower-bounds} we prove this Ramsey-theoretic lemma.

\begin{restatable}{lem}{almostpositive}\label{almost-positive-cylinder-dependence}
	Let $r\in \mathbb{N}$ and $0\leq \epsilon\leq p \leq 1$. Then, there exists some $n = n_r(p,\epsilon)$ such that any Bernoulli variables $Y_1,\dots,Y_n\sim \Ber(p)$ contains some subset $I\subseteq [n]$ of size $|I|=2^r$ such that 
	$$\E\left[\prod_{i} Y_i\right] \geq \prod_i \E[Y_i] - \epsilon.$$ 
\end{restatable}
    \section{Summary and Open Questions}
In this work we provided a methodical study of online dependent rounding schemes (ODRSes) for bipartite $b$-matching. We obtained optimal rounding ratios with strong concentration for star graphs (uniform-matroid rounding) and improved bounds beyond a ratio of $1-\frac{1}{e}$ for ($b$-)matchings, including the first results for $b$-matchings not obtained via OCRSes. (Consequently, ours is the only ratio beyond $1/2$ known for this problem.)
We furthermore provided a number of applications of our ODRSes and their algorithmic approach.
Beyond the natural question of improving our rounding ratios for $b$-matchings (and the approximation/competitive ratios of their derived applications), our work raises a number of directions for future research.

\paragraph{Online applications.} We provided a number of online applications of our ODRSes.
What further applications do these schemes have? The bipartite setting seems particularly relevant for online scheduling problems, with machines and jobs represented by offline and online nodes, respectively.
For minimization problems, it may prove useful to observe that our ODRSes for $b$-matching have modest two-sided error: edges $(i,t)$ are matched with probability at most the probability that $i$ bids for $t$, namely $\hat{x}_{i,t}\leq (1+\delta)x_{i,t} \approx 1.064\cdot x_{i,t}$. 
Moreover, these latter events satisfy strong upper-tail bounds, by our ODRS' strong negative correlation properties (see Property \ref{prop:sr} and \Cref{NA-S_it}).

\paragraph{Streaming applications.}
We note that our ODRSes are not only online algorithms, but moreover are streaming (for uniform matroids) or semi-streaming algorithms (for arbitrary $b$-matchings). Indeed, they require us to only remember $O(\log n)$-bit partial sums and a single number per offline node. 
Does this property have further applications in streaming or graph-analytics contexts?

\paragraph{Beyond matchings.}
Finally, we recall that oblivious online contention resolution schemes (OCRSes) yield ODRSes, though this connection does not yield optimal rounding ratios.
Recent years have seen an explosion of work on OCRSes for increasingly-rich families of arrival models and combinatorial constraints. These have been fueled in large part by connections to various other online and economic applications, most notably the prophet-inequality problem under more involved combinatorial constraints (see discussion in the influential work of \citet{kleinberg2012matroid}).
Can a similarly rich theory of online dependent rounding more broadly be developed, capturing other combinatorial constraints beyond matchings?

\color{black}

	\subsection*{Acknowledgements.} 
    We thank Uri Feige for pointing out the connection of our work to offline CRSes, and thank Sahil Singla for discussions about (lack of prior) examples of black-box application of offline CRSes to online problems.
	We thank Manish Purohit for asking about the two-sided error of our ODRSes, which we anticipate may be of use in future applications for online minimization problems.
    Finally, we thank the anonymous reviewers for comments on presentation.
	
	\appendix
    \section*{APPENDIX}
    \section{Additional Preliminaries}\label{appendix:prelims}

In this section we provide omitted proofs and additional preliminaries not covered in \Cref{sec:prelims}.

\subsection{Contention Resolution Schemes}
We start with a brief, self-contained proof of \Cref{lem:CRS}, restated below for ease of reference.
\CRS*
\begin{proof}
    The lemma for product distributions follows from the work of \citet{feige2006allocation}. We focus on the more general case, addressed by \citet{bansal2021contention}, but proven here for completeness.

    Let $\alpha := \min_{S\subseteq [n]}\frac{\Pr[R\cap S \neq \emptyset]}{\sum_{i\in S} v_i}$.
    We prove the existence of conditional probabilities $p_{i,S}$ with $p_{i,S}=0$ if $i\not\in S$ and $\sum_i p_{i,S}=1$ for all $S$ such that setting $\Pr[O = \{i\} \mid R = S]=p_{i,S}$ yields $\Pr[O = \{i\}] = \sum_{S} \Pr[O = \{i\} \mid R = S]\cdot \Pr[R = S] \geq \alpha \cdot v_i$.
	We prove the existence of these probabilities using the max-flow/min-cut theorem applied to the following directed flow network that has a source vertex $src$, a sink vertex $sink$, and two other disjoint sets of vertices $A$ and $B$. Vertices on side $A$ correspond to subsets of elements in $[n]$, with an edge of capacity $\Pr[R=S]$ from $src$ to the node in $A$ representing set $S$. Vertices on side $B$ correspond to elements $i\in [n]$, with an edge of capacity $\alpha \cdot v_i$ leaving this vertex to $sink$. Infinite-capacity directed edges go from each vertex  in $A$ corresponding to some set $S$, to vertices in $B$ corresponding to each $i\in S$. Each vertex $i\in S$ in $B$, has a directed edge to $sink$ with capacity $\alpha\cdot v_i$. 
	This network is reminiscent of the network used to prove Hall's Theorem via the max-flow/min-cut theorem.
	Indeed, $\min_{S\subseteq [n]}\frac{\Pr[R\cap S \neq \emptyset]}{\sum_{i\in S} v_i}\geq \alpha$ is a (scaled) version of Hall's condition, and implies by the same argument the existence of a ``perfect'' flow $f$ of value $\alpha \cdot \sum_i v_i$, implying that each edge of the form $(i, sink)$ has its capacity $\alpha\cdot v_i$ saturated by this flow.
	This in turn implies the existence of the probabilities $p_{i,S}$ as desired, where if the flows through the edges  $(src,S)$ and $(S,i)$ are $f_{src, S}$ and $f_{S,i}$ respectively, then $p_{i,S} = f_{S,i}/\Pr[R = S]$ (if $\Pr[R = S] = 0$, we take arbitrary non-negative $p_{i,S}$ with $p_{i,S}=0$ if $i\not\in S$ and $\sum_{i: i \in S} p_{i,S}=1$). Thus, the above algorithm outputs a set $O\subseteq R$ of size at most one with each element belonging to $O$ with the requisite probability.
	The running time of the algorithm for general distributions follows by polytime solubility of the maximum flow problem.
\end{proof}

\subsection{Negative correlation properties}
Here we formalize the strong notion of negative correlation that our online level-set algorithm enjoys.

\begin{Def}
	A distribution $\mu:2^{[n]}\to \mathbb{R}_{\geq 0}$ is \emph{strongly Rayleigh} if its generating polynomial,
	$$g_\mu(z) = \sum_{S\subseteq [n]} \mu(S)\prod_{i\in S} z_i,$$
	is \emph{real stable}, i.e., it has no root $z\in \mathbb{C}^n$ satisfying $\Im(z_i)>0$ for all $i\in [n]$.
\end{Def}

As stated in \Cref{sec:prelims} and proven in \cite{borcea2009negative}, the strong Rayleigh property implies NA.

\begin{lem}
    If $(X_1,\dots,X_n)$ are strongly Rayleigh binary r.v.s, then they are NA.
\end{lem}
Moreover, since SRP is closed under conditioning \cite{borcea2009negative}, we find that SRP distributions are NA even after conditioning.

The strong Rayleigh property, central to the area of geometry of polynomials, has been highly influential in recent years.
To the best of our knowledge, aside from the classic balls-and-bins process, its variants and conditional counterparts (see \cite{branden2012negative}), our online level-set algorithm is the only online algorithm known to satisfy such a strong negative-correlation property. 
We believe that this property of our algorithm will find further applications outside of this work.

\paragraph{Submodular dominance.} As stated in \Cref{sec:prelims}, negative association implies stochastic dominance in the submodular order. To formally define the above, we briefly recall the definition of submodular functions, which are set functions capturing the notion of diminishing returns.

\begin{Def}
	A set function $f:2^{[n]} \to \mathbb{R}$ is \emph{submodular} if for all sets $A\subseteq B\subseteq [n]$ and element $x\in [n]\setminus B$, we have that
	$$f(A\cup \{x\}) - f(A) \geq f(B\cup \{x\}) - f(B).$$
\end{Def}

The following submodular dominance result of \citet{christofides2004connection}, later generalized by Singla and Qiu \cite{qiu2022submodular} implies that our online level-set algorithm not only preserves weighted objectives losslessly, but also preserves submodular objectives. (See \Cref{appendix:level-sets}.)

\begin{thm}\label{thm:stochastic-order}
	Let $X_1,\dots,X_n$ be binary NA random variables and let $X^*_1,\dots,X^*_n$ be independent random variables that are place-wise equal to the $X_1,\dots,X_n$ respectively in distribution.
	Then, for any monotone submodular function $f$, 
	$$\E[f(X_1,\dots,X_n)] \geq \E[f(X^*_1,\dots,X^*_n)].$$
\end{thm}

Moreover, the following result of Duppala et al.~\cite{duppala2023concentration} even implies concentration of the above.

\begin{thm}\label{submod-concentration}
    	Let $X_1,\dots,X_n$ be binary NA random variables and let $X^*_1,\dots,X^*_n$ be independent random variables that are place-wise equal to the $X_1,\dots,X_n$ respectively in distribution.
     Then, for any monotone submodular function $f$, with $\mu:=\E[f(X^*_1,\dots,X^*_n)]$ and real value $\delta>0$, 
$$\Pr[f(X_1,\dots,X_n)\leq (1-\delta)\cdot \mu]\leq \exp\left(-\frac{\delta^2\mu}{2}\right).$$
\end{thm}

\subsection{Odds and Ends}
We will also make use of the following direct corollary of convexity.
\begin{claim}\label{convexity}
	Let $f$ be a non-negative concave function. Then, for any $0\leq t\leq \alpha$, 
	$$\frac{f(t)}{t}\geq \frac{f(\alpha)}{\alpha}.$$
\end{claim}
\begin{proof}
	Writing $t$ as a convex combination of $\alpha$ and $0$, the claim follows from convexity, as follows.
	\begin{align*}
		f(t) & = f\left(\frac{t}{\alpha}\cdot \alpha + \left(1-\frac{t}{\alpha}\right)\cdot 0\right) \geq \frac{t}{\alpha}\cdot f(\alpha) + \left(1-\frac{t}{\alpha}\right)\cdot f(0) \geq \frac{t}{\alpha}\cdot f(\alpha). \qedhere
	\end{align*}
\end{proof}

	\section{Deferred Proofs of \Cref{sec:rounding-matchings}}\label{appendix:b-matching}

In this section we provide deferred proofs of claims from \Cref{sec:rounding-matchings}, restated below for ease of reference. 

\minimizer*
\begin{proof}
%[Proof of \Cref{claim:miminimzer}]
	Taking the derivative of $g(y)$, we get $$g'(y)=\exp(-(1-y)\cdot (1+\delta))\cdot \left((1-\eps)-(1-y\cdot (1-\eps))\cdot (1+\delta)\right),$$
	which vanishes at $y^*=\frac{\eps+\delta}{(1-\eps)(1+\delta)}$.
    Next, we consider the second derivative of $g$.
    $$g''(y)=g'(y)\cdot (1+\delta) + \exp(-(1-y)\cdot (1+\delta))\cdot (1-\eps)\cdot (1+\delta).$$
    The first summand, $g'(y)\cdot (1+\delta)$, vanishes at $y^*$, by the above, while the second summand is positive, since $\eps\leq 1$ and $\delta\geq 0$.
    We conclude that $y^*$ is the global minimum of $g$.
\end{proof}

\sufficient*
\begin{proof}
	Fix $\epsilon,\delta$. For notational simplicity, let $f=f_{\epsilon,\delta}$ and let $z^*=z^*_{\epsilon,\delta}$.
	Since $f$ is the sum of twice-differentiable convex functions (namely an exponential and a linear term), we have that $f''(z)\geq 0$ for all $z$. Therefore, if $f'(z^*)\geq 0$, then for all $z\geq z^*$ we have that
	\begin{align*}
		f'(z) = f'(z^*)+\int_{z^*}^z f''(x)\,dx \geq f'(z^*) \geq 0.
	\end{align*} 
	Thus, if moreover $f(z^*)\geq 0$, then the desired statement follows, i.e., for all $x\geq z^*$ we have that
	\begin{align*}
		f(z) &= f(z^*) + \int_{z^*}^z f'(x) \, dx \geq f(z^*) \geq 0. \qedhere 
	\end{align*}
\end{proof}

\subsection{Polytime implementation}\label{sec:polytime}
So far, we ignored computational aspects of our ODRS of \Cref{thm:ODRS}. In this section we show how our grouping method used to increase our rounding ratio beyond $1-\frac{1}{e}$ moreover allows us to compute $\Pr[P_t=S]$ for all sets $S\subseteq [n]$ in polynomial time --- and thus obtain a polytime ODRS --- while only decreasing our rounding ratio by an arbitrarily small constant $\gamma>0$.

\polymatching*
\begin{proof}
    First, we scale down all the fractions $x_{i,t}$ by a $(1-\gamma)$ factor (this results in a valid fractional matching).
    As this additional downscaling is linear, this results in each offline node having fractional degree $\hat{s}_{i,t}\leq 1-\gamma$ in the assignment $\hat{x}$, by \Cref{hat-degrees}.
    Consequently, by the same argument as \Cref{first-fit}, each bin but at most one in $B\in \calB_t$ has $\sum_{i\in B} \hat{x}_{i,t}\leq \frac{\gamma}{2}$.
    But since $\sum_i \hat{x}_{i,t}\leq \sum_i x_{i,t}\cdot (1+\delta)\leq (1+\delta)$, this implies that each bin but one has total $x$-value at least $\sum_{i\in B} x_{it} \geq \gamma/2(1+\delta)$. 
    Consequently, by the fractional matching constraint, whereby $\sum_i x_{i,t}\leq 1$, 
    there are at most $2(1+\delta)/\gamma+1=O(1/\gamma)$ such non-empty bins.
    
    The benefit of few bins is that now at each time $t$ the set of bidders is a set of at most $O(1/\gamma)$ offline nodes, and so there are at most $n^{O(1/\eps)}$ possible sets $S\subseteq [n]$ in the support of the distribution over $P_t$.
    Moreover, for each such set $S\subseteq [n]$, it is straightforward to compute the probability that $S$ is the set of first-time proposers at time $t$ in time $n^{O(1/\gamma)}$. (Briefly, for each subset $S$, we guess for each of the $O(1/\gamma)$ nodes in $S$ at what time they made their first bid. The probability of each of these $n^{O(1/\gamma)}$ guessed events is then easy to compute in polynomial time, by computing the probability that these nodes did bid in these guessed time steps and no other times.) All in all, this allows to compute, exactly, the support of the distribution of $P_t$, in time $n^{O(1/\gamma)}$. 
    Therefore, by \Cref{lem:CRS}, the CRS invocations take time $\poly(n^{O(1/\gamma)}=n^{O(1/\gamma)}$ per time step.
    On the other hand, having scaled down the values $x_{i,t}$ by  a factor of $(1-\gamma)$ trivially incurs a multiplicative loss of $(1-\gamma)$ in the rounding ratio, and so the obtained ODRS has rounding ratio of $\ratio(1-\gamma)\geq \ratio-\gamma$.
\end{proof}

\subsection{The $b$-matching ODRS}\label{sec:b-matching-extension}
In this section we extend our ODRS from matchings to $b$-matchings.
To avoid repetition, we only outline the changes needed in the algorithm and analysis, and defer the proofs to \Cref{app:b-matching-gen}.

\paragraph{The change to the core algorithm.}
As we did (implicitly) in \Cref{alg:bucketing}, we let the online level-set algorithm dictate the marginal probabilities with which an offline node $i$ should bid to $t$, based on the number of times $i$ has bid by time $t$, which we denote by $S_{i,t}$. (This corresponds to $S_t$ in \Cref{alg:rounding} run from the perspective of the offline node $i$.) We process offline nodes at different times $t$ as follows, with the choice of probabilities guided by \Cref{alg:rounding} to guarantee that $\E[S_{i,t+1} - S_{i,t}] = x_{i,t}$, by Property \ref{prop:marginals}, and $S_{i,t}\in [\lfloor s_{i,t} \rfloor, \lceil s_{i,t}\rceil]$, by Property \ref{prop:floor/ceil}.

For nodes with $\lceil s_{t+1} \rceil = \lceil s_{t} \rceil$ --- whose fractional degrees both before and after time $t$ lie in the range $[\lfloor s_{i,t} \rfloor, \lceil s_{i,t} \rceil]$ --- we group these offline nodes and have them bid (at most one bid per bin) as in \Cref{alg:bucketing}, with the following changes: $s_{i,t}$ is replaced by its fractional part, $s_{i,t} - \lfloor s_{i,t} \rfloor$: this is done both in the bidding probability, so these items have size $\frac{x_{i,t}}{1-(s_{i,t}-\lfloor s_{i,t} \rfloor)}$ when bin packing, and in the condition $s_{i,t}\leq \theta$, which are now replaced by the condition $s_{i,t} - \lfloor s_{i,t}\rfloor \leq \theta$ (and similarly for the condition $s_{i,t}>\theta$), with $\theta$ to be chosen later.
Moreover, we add a candidate offline node $i\in C_t$ to the list of potential matches $P_t$ (the bidders) if $S_{i,t} < \lfloor s_{i,t} \rfloor$ (corresponding to $i\in F\cap C_t$ in \Cref{alg:bucketing} for simple matchings).

For all offline nodes with $\lfloor s_{i,t+1} \rfloor = \lfloor s_{i,t} \rfloor + 1$, we place these nodes in singleton bins, and we have them bid with probability if $S_{i,t} = \lfloor s_{i,t}\rfloor$, or with probability $\frac{s_{i,t+1}-\lfloor s_{i,t+1} \rfloor}{s_{i,t}-\lfloor s_{i,t} \rfloor}$ if $S_{i,t} = \lceil s_{i,t} \rceil = \lfloor s_{i,t+1}\rfloor$. 

\paragraph{Change to core ODRS' analysis:} \Cref{first-fit} and its proof hold unchanged for our extension of the core ODRS. \Cref{proposal-prob-perturbed-rephrased}, too, holds unchanged, although the proof, deferred to \Cref{app:b-matching-gen}, is now slightly more involved. Specifically, the key ingredient, proving that the events $[P_t\cap S\cap B]$ over different bins $B$ are negatively associated (NA) requires a proof that the variables $S_{i,t}$ are NA.

\paragraph{Change to the improved ODRS.}
In our ODRS for $b$-matchings, since offline nodes $i$  for which $\lfloor s_{i,t+1} \rfloor = \lfloor s_{i,t} \rfloor + 1$ are grouped into singelton bins, we wish to apply a markup to these nodes' $x_{i,t}$-values, if $x_{i,t}$ is small. We therefore consider the following assignment, with $0\leq \theta_1\leq \theta_2\leq 1$ satisfying 
$\theta_2-\theta_1=\frac{\delta}{\eps+\delta}$, 
 to be chosen shortly.
\begin{align}\label{eqn:x-hat-b-matching}
    \hat{x}_{i,t} = \int_{z=s_{i,t}}^{s_{i,t+1}} (1-\eps)\cdot \mathds{1} \left[\lfloor z\rfloor \in [\theta_1,\theta_2) \right] + (1+\delta) \cdot \mathds{1} \left[\lfloor z\rfloor \not\in [\theta_1,\theta_2) \right] \, dz.
\end{align}
In words: this assignment decreases the part of $x_{i,t}$ corresponding to the parts of $i$'s fractional degree
that are bounded away from $0$ and $1$ (specifically, those in the range $[\theta_1,\theta_2)$, and increases the part of the fractional degree that is somewhat close to $0$ or $1$, in the ranges $[0,\theta_1)$ and $[\theta_2,1]$. 
Our new ODRS then runs the modified core ODRS with inputs $\theta := \theta_2\cdot (1-\eps) + \theta_1\cdot (\delta+\eps)$ and $\mathbf{\hat{x}},\mathbf{x}$. The choice of the parameter $\theta$ is taken to have $\hat{s}_{i,t} - \lfloor \hat{s}_{i,t}\rfloor \leq \theta$ if and only if ${s}_{i,t} - \lfloor {s}_{i,t}\rfloor \leq \theta_2$.

\paragraph{Change to the improved ODRS' analysis.} First, \Cref{hat-degrees} generalizes as follows, arguing that modified fractional degrees $\hat{s}_{i,t}=\sum_{t'<t}\hat{x}_{i,t'}$ fall within the same integer range for all $(i,t)$.
\begin{fact}\label{hat-degrees-b-matching}
For every offline node $i\in [n]$ and time $t$, we have that 
$\lfloor \hat{s}_{i,t}\rfloor = \lfloor s_{i,t} \rfloor$ and 
$\lceil \hat{s}_{i,t}\rceil = \lceil s_{i,t} \rceil.$ 
\end{fact}
\Cref{hat-degrees-b-matching} and Property \ref{prop:floor/ceil} of \Cref{alg:rounding} imply that offline node $i$ makes some $S_{i,t}\in [\lfloor s_{i,t}\rfloor, \lceil s_{i,t} \rceil]$ many bids by time $t$. As $i$ is not matched more time than it bids, this implies that the extended ODRS outputs a valid $b$-matching.

Next, the key lemma in the analysis of \Cref{alg:ODRS}, namely \Cref{conditional-ratio}, needs to be extended slightly, and in particular the condition of \Cref{eqn:conditional-ratio} needs to hold for all $z\geq \max\left\{\frac{1-\theta_2}{2(1+\delta)},\theta_1\right\}$. 
To make this constraint as less restrictive as possible, we take $\theta_1$ and $\theta_2$ to guarantee the needed $\theta_2 - \theta_1 = \frac{\delta}{\eps+\delta}$, and also satisfy $\frac{1-\theta_2}{2(1+\delta)} = \theta_1$, resulting in $\theta_1 = \frac{\eps}{(3+2\delta)(\eps+\delta)}$ and $\theta_2 = 
\frac{\eps+3\delta+2\delta^2}{(3+2\delta)(\eps+\delta)}$.
This modification of the range for which we require Condition \eqref{eqn:conditional-ratio} to hold will again allow us to argue that for all bins $B$ except for at most one bin $B^*$ (to be characterized shortly), we have the following: 
\begin{align}\label{eqn:use-of-condition}
1-\sum_{i\in B}\hat{x}_{i,t} \leq \exp\left(-\sum_{i\in B}x_{i,t}(1+\delta)\right)
\end{align}
Specifically, for bins containing a single offline node $i$ with $\lfloor s_{i,t+1} \rfloor = \lfloor s_{i,t} \rfloor + 1$, we have the following: 
If $s_{i,t+1}-\lfloor s_{i,t+1} \rfloor \geq \theta_1$, then  
$$x_{i,t} = s_{i,t+1} - s_{i,t} \geq s_{i,t+1}-\lfloor s_{i,t+1} \rfloor\geq \theta_1.$$
Similarly, if $s_{i,t}-\lfloor s_{i,t+1} \rfloor \leq \theta_2$, then 
$$x_{i,t} = s_{i,t+1} - s_{i,t} \geq \lfloor s_{i,t+1}\rfloor - s_{i,t} = 1 + \lfloor s_{i,t} \rfloor - s_{i,t} \geq 1- \theta_2 \geq \frac{1-\theta_2}{2(1+\delta)}.$$
In both cases, Condition \eqref{eqn:conditional-ratio} holding for all $z\geq \max\left\{\frac{1-\theta_2}{2(1+\delta)},\theta_1\right\}$ implies \Cref{eqn:use-of-condition}.
In the alternative scenario where $s_{i,t+1}-\lfloor s_{i,t+1} \rfloor \leq \theta_1$ and $s_{i,t}-\lfloor s_{i,t} \rfloor \leq \theta_2$, we have that $\hat{x}_{i,t} = (1+\delta)\cdot x_{i,t}$, and \Cref{eqn:use-of-condition} follows from the basic inequality $1-z\leq \exp(-z)$.
Finally, by \Cref{first-fit}, for all bins $B$ containing offline nodes $i$ with $\lfloor s_{i,t} \rfloor = \lfloor s_{i,t+1} \rfloor$ and $\hat{s}_{i,t}\leq \theta$ (i.e., $s_{i,t}\leq \theta_2$), except for at most one bin $B^*$, we have that $$\sum_{i\in B}{x}_{i,t}\cdot (1+\delta)\geq \sum_{i\in B}\hat{x}_{i,t}\geq \frac{1-\theta}{2} \geq \frac{1-\theta_2}{2},$$
where the last inequality, whereby $\theta\leq \theta_2$, is proven similarly to \Cref{hat-degrees-b-matching}.
So, Condition \eqref{eqn:conditional-ratio} holding for all $z\geq \frac{1-\theta_2}{2(1+\delta)}$ implies \Cref{eqn:use-of-condition} for all remaining bins other than $B^*$.
The analysis then proceeds as in \Cref{conditional-ratio}, yielding the following.

\begin{lem}\label{conditional-ratio-b-matching}
	If $\epsilon,\delta \in [0,1]$ guarantee that for all $z\geq \max\left\{\frac{1-\theta_2}{2},\,\theta_1\right\} = \frac{\epsilon}{\mathbf{(3+2\delta)}\cdot (\epsilon+\delta)}$, the expression 
	\begin{align*}
		f_{\epsilon,\delta}(z) := \exp(-z\cdot (1+\delta)) - (1-z\cdot (1-\eps))
	\end{align*}
    is non-negative, then the extension of \Cref{alg:ODRS} to $b$-matchings using scaling \Cref{eqn:x-hat-b-matching} and parameters $\epsilon$ and $\delta$ has rounding ratio at least 
	$$1-\exp\left(-1-\delta+\frac{\eps+\delta}{1-\eps}\right)\cdot\frac{1-\eps}{1+\delta}.$$
	%$$1-\exp\left(-1+\frac{\epsilon}{2}-\delta\right).$$
\end{lem}

Combined with \Cref{sufficient-condition}, the above lemma then yields our $b$-matching ODRS' rounding ratio, after solving the appropriate modification of the mathematical program used to prove \Cref{thm:ODRS}.

\begin{thm}\label{thm:ODRS-b-matching}
    The extension of \Cref{alg:ODRS} to $b$-matchings, run with $(\eps,\delta)\approx (0.0347,0.0425)$, achieves a rounding ratio of $\ratiobmatching$.
\end{thm}

% Mathematica code:

% x[eps_, delta_] := eps/((3 + 2*delta)*(eps + delta))
% NMaximize[{1 - (1 - ((eps + delta)/((1 + delta)*(1 - eps)))*(1 - 
%          eps))*Exp[-(1 - (eps + delta)/((1 + delta)*(1 - eps)))*(1 + 
%         delta)], 
%   Exp[-x[eps, delta]*(1 + delta)] - 1 + x[eps, delta]*(1 - eps) >= 
%    0, -(1 + delta)*Exp[-x[eps, delta]*(1 + delta)] + 1 - eps >= 0,
%   0 <= eps <= 1, 0 <= delta <= 1}, {eps, delta}]

% Expected output:
% {0.646023, {eps -> 0.0338701, delta -> 0.0414096}}

\subsection{Deferred proofs of the $b$-matching ODRS}\label{app:b-matching-gen}

In this section we prove that our extension of \Cref{alg:bucketing} from matching to $b$-matchings results in similar benefits when much grouping occurs. We start by introducing some useful notation.

    Our online level-set rounding algorithm, \Cref{alg:rounding}, makes at most one random choice when deciding whether to allocate another online node to the center of the star.
    Our $b$-matching ODRS follows the same logic (from the perspective of the offline node, playing the role of the star's center).
    Denote by $C_{i,t}$ an indicator for the relevant coin for $i$ coming up heads. Our extended core ODRS, \Cref{alg:bucketing}, correlates these coins for different offline nodes $i$ at time $t$.
    Let $S_{i,t}\in \{\lfloor s_{i,t} \rfloor, \lceil s_{i,t} \rceil\}$ be the number of times $i$ bids by time $t$. By our algorithm's definition, we have the following expression for $[S_{i,t} \neq \lceil s_{i,t} \rceil]$.
    \begin{align*}
        [S_{i,t} \neq \lceil s_{i,t} \rceil] = \begin{cases}
        [S_{i,t-1} \neq \lceil s_{i,t-1} \rceil]\cdot (1-C_{i,t-1}) & \textrm{if } \lceil s_{i,t} \rceil = \lceil s_{i,t-1} \rceil \\
        1 - [S_{i,t-1} = \lceil s_{i,t-1} \rceil]\cdot C_{i,t-1} & \textrm{if }  \lceil s_{i,t} \rceil > \lceil s_{i,t-1} \rceil.
        \end{cases}
    \end{align*}

\begin{lem}\label{NA-S_it}
    For any time $t$, the random variables $\{S_{i,t}\}_i$ are NA.
\end{lem}
\begin{proof}
    Proof by induction on $t\geq 1$. The base case is trivial, as $S_{i,1}=0$  deterministically for all $i$. 
    We turn to the inductive step for $t$, assuming that the variables $S_{i,t-1}$ are NA.
    First, by \Cref{lem:0-1}, for any bin $B\in \calB_{t-1}$ at time $t-1$, the binary variables $C_{i,t-1}$ for all nodes $i\in B$, whose sum is at most one, are NA.
    Moreover, the distributions over different bins are independent, and so by closure of NA under products (\Cref{na-closure}), all $C_{i,t-1}$ are NA.
    Moreover, the variables $C_{i,t-1}$ are independent from the variables $S_{i,t-1}$ (who are functions only of $C_{i,t'}$ for all $t'<t-1$), and consequently, by our inductive hypothesis and another application of closure of NA under products (\Cref{na-closure}), all the variables $\{S_{i,t-1},C_{i,t-1}\}_i$ are NA.
    Therefore, since $S_{i,t-1}\in [\lfloor s_{i,t-1} \rfloor, \lceil s_{i,t-1} \rceil]$, we find that $[S_{i,t}\neq \lceil s_{i,t}\rceil]$ is a decreasing function of $S_{i,t-1}$ and $C_{i,t-1}$, i.e., they are decreasing functions of disjoint NA variables.
    Therefore, by yet another application of closure of NA under products (\Cref{na-closure}), the events $\{[S_{i,t}\neq \lceil s_{i,t} \rceil]\}_i$ are NA.
    Finally, since $S_{i,t} = \lceil s_{i,t} \rceil - [S_{i,t}\neq \lceil s_{i,t} \rceil]$ are monotone decreasing functions of NA variables, the variables $S_{i,t}$ are themselves NA, as desired.
\end{proof}

We are now ready to prove that our generalization of \Cref{alg:bucketing} results in bins essentially simulating offline nodes bidding with a probability equal to their constituent nodes' bids. That is, we are ready to prove the following generalization of \Cref{proposal-prob-perturbed-rephrased}.
\begin{lem}\label{proposal-prob-perturbed-b-matching}
	For any time $t$ and set of offline nodes $S\subseteq[n]$, the $b$-matching extension of \Cref{alg:bucketing} satisfies
	$$\Pr[S\cap P_t \neq \emptyset] \geq 1 - \prod_{B\in \calB_t}\left(1-\sum_{i\in B\cap S} x_{it}\right).$$
\end{lem}
\begin{proof}
    We denote by $\calB'_t\subseteq \calB_t$ the set of singleton bins $B=\{i\}$ for which $\lceil s_{i,t+1} \rceil > \lceil s_{i,t} \rceil$.
    For such bins $B\in \calB'_t$, $B=\{i\}$, we have that $\Pr[P_t\cap S\cap B = \emptyset] = 1-x_{i,t}$.
    For other bins $B\in \calB_t\setminus \calB'_t$, we have that $\Pr[i\in P_t] = x_{i,t}$, by independence of $C_{i,t}$ and $S_{i,t}$ 
    and since $\Pr[C_{i,t}]=\frac{x_{i,t}}{1-(s_{i,t}-\lfloor s_{i,t}\rfloor)}$ and $\Pr[S_{i,t}=\lfloor s_{i,t} \rfloor] = s_{i,t}-\lfloor s_{i,t}\rfloor$ 
    (the latter following from our online level-set algorithm's analysis). Therefore, by $\sum_{i\in B}C_{i,t}\leq 1$ implying disjointness of the events $[i\in P_t]$ for $i\in B$, we have that $\Pr[P_t \cap S\cap B = \emptyset] = 1-\sum_{i\in B\cap S} x_{i,t}$. It remains to show that the events $[P_t \cap S\cap B = \emptyset]$ for different $B$ are negative cylinder dependent, giving the desired inequality,
    \begin{align*}
    \Pr[P_t\cap S = \emptyset] = \Pr\left[\bigwedge_{B\in \calB_t} (P_t \cap S\cap B = \emptyset)\right] \leq \prod_{B\in \calB_t} \Pr[P_t \cap S\cap B = \emptyset] = \prod_{B\in \calB_t} \left(1-\sum_{i\in B\cap S} x_{i,t}\right).
    \end{align*}
    Now, for bins $B=\{i\}\in \calB'_t$, we have that $\mathds{1}[P_t\cap S\cap B = \emptyset] =1 - [S_{i,t} = \lceil s_{i,t} \rceil]\cdot C_{i,t}$.
    For non-singleton bins $B\in \calB_t$, we have that
    $\mathds{1}[P_t \cap S\cap B = \emptyset] = 1-\sum_{i\in B\cap S} C_{i,t}\cdot [S_{i,t}=\lfloor s_{i,t}\rfloor]$. 
    In both cases, the indicators $\mathds{1}[P_t \cap S\cap B = \emptyset]$ are NA, as they 
    are monotone decreasing functions of disjoint NA variables (here, we need to first apply decreasing functions that map $S_{i,t}$ to $\lceil s_{i,t}\rceil - S_{i,t}$ for all $i\in B\in \calB_t\setminus \calB'_t$ and all other variables to themselves). Hence, by closure of NA under such functions (\Cref{na-closure}), the variables $\mathds{1}[P_t \cap S\cap B = \emptyset]$ are NA.
    Therefore, by \Cref{na-implies-neg-cylinder}, these variables are negative cylinder dependent,
    yielding the required inequality above. The lemma follows.   
\end{proof}
    
\section{Deferred Proofs of \Cref{sec:level-sets}}\label{appendix:level-sets}

We start by briefly proving that \Cref{alg:linear-srinivasan-rounding} terminates and preserves marginals.

\begin{fact}
\Cref{alg:linear-srinivasan-rounding} terminates, and satisfies $\Pr[i\in \calS]=x_i$ for all $i\in [n]$ (Property \ref{prop:marginals}).
\end{fact}
\begin{proof}
It is easy to see that STEP preserves the sum of its inputs/outputs while decreasing the number of fractional such inputs/outputs. Consequently, $\sum_i y_i$ is integral during the execution of \Cref{alg:linear-srinivasan-rounding}, while the number of fractional $y_i$ decreases and cannot be one, so this algorithm terminates with all $y_i$ binary.
Moreover, STEP is easily seen to preserve these values' expectations, and so by induction we have that $\E[y_i] = x_i$ for all $i$, i.e., \Cref{alg:linear-srinivasan-rounding} satisfies Property \ref{prop:marginals}. 
\end{proof}

\offlineprefix*
\begin{proof}
	We prove this for all $j$ by induction on $t\geq 0$. The base case is trivial. Fix a time $t$ and let $i_1$ and $i_2$ be as in \Cref{alg:linear-srinivasan-rounding} at time $t+1$. If $i_1,i_2\not\in [j]$ or $i_1,i_2\in [j]$, then, since STEP$(y_{i_1},y_{i_2})$ does not affect $Y^t_j$ in the former case and preserves the sums $y_{i_1}+y_{i_2}$ and $Y_j$ in the latter case, the inductive hypothesis implies $Y^{t+1}_j = Y^t_j \in \left[\lfloor s_j \rfloor, \lceil s_j \rceil \right]$. Conversely, if $i_1\in [j]$ and $i_2\not\in [j]$ (and therefore $i_1=j$), we have that $y^{t+1}_i = y^t_i\in \{0,1\}$ for all $i\in [j-1]$, while $y^{t+1}_{i_1}\in [0,1]$. Since by the inductive hypothesis we have that $Y^t_j = Y^t_{j-1} + y^t_j \in \left[\lfloor s_j \rfloor, \lceil s_j \rceil \right]$ and $Y^t_{j-1}=Y^{t+1}_{j-1}$ is the sum of 0-1 terms, while $y^t_j\in (0,1)$, this implies that $Y^{t+1}_j = Y^{t+1}_{j-1} + y^{t+1}_j = Y^{t}_{j-1} + y^{t+1}_j\in \left[\lfloor s_j \rfloor, \lceil s_j \rceil \right]$.
\end{proof}

\subsection{Direct (partial) analysis of \Cref{alg:rounding}}In this section we provide a simple, self-contained proof that \Cref{alg:rounding} satisfies the first two properties we desire of it. 
In the paper body, we use this direct proof to prove that our algorithm's output distribution is the same as its offline counterpart, from which we also obtain the third desideratum, namely Property \ref{prop:sr}.

\begin{lem}\label{level-set:simple}
	\Cref{alg:rounding} is well-defined (i.e., $p_t\in [0,1]$ for all times $t$ and any realization of the randomness) and it satisfies properties \ref{prop:marginals} and \ref{prop:floor/ceil}.
\end{lem}
\begin{proof}
	We prove the above three properties (including $p_t\in [0,1]$) by strong induction on $t\geq 0$. The base case, $t=0$, is trivial. As our inductive hypothesis (I.H.), we assume that these properties hold for all $t'\leq t-1$ and prove them for all $t'\leq t$, starting with Property \ref{prop:floor/ceil}. By the I.H., $S_{t-1}\leq \lceil s_{t-1} \rceil \leq \lceil s_t \rceil$; so, since $p_t=0$ if $S_{t-1} = \lceil s_{t}\rceil$, we trivially have that $S_{t}\leq \lceil s_{t} \rceil$. 
	Moreover, by the I.H., $S_{t-1} \geq \lfloor s_{t-1} \rfloor = \lfloor s_t - x_{t} \rfloor \geq  \lfloor s_{t} \rfloor - 1$; so, since $p_{t}=1$ if $S_{t-1}<\lfloor s_{t} \rfloor$, in which case $S_{t-1} = \lfloor s_t \rfloor - 1$, we have that $S_{t}\geq \lfloor s_{t} \rfloor$. 
	Thus, $S_{t} \in \{\lfloor s_{t} \rfloor, \lceil s_{t} \rceil\}$, proving Property \ref{prop:floor/ceil}.
	
	Next we prove that \Cref{alg:rounding} is well-defined, as $p_t\in [0,1]$ (regardless of prior randomness). 
	The first non-trivial case is when $S_{t-1}=\lfloor s_{t} \rfloor = \lfloor s_{t-1} \rfloor$. Non-negativity of $p_t$---a ratio of non-negative terms---is immediate, while $$p_t = \frac{x_t}{\lfloor s_{t-1} \rfloor + 1 - s_{t-1}} = \frac{s_t - s_{t-1}}{\lfloor s_{t-1} \rfloor + 1 - s_{t-1}} \leq \frac{\lfloor s_t \rfloor + 1 - s_{t-1}}{\lfloor s_{t-1} \rfloor + 1 - s_{t-1}} = \frac{\lfloor s_{t-1} \rfloor + 1 - s_{t-1}}{\lfloor s_{t-1} \rfloor + 1 - s_{t-1}} = 1.$$ The second non-trivial case is when $S_{t-1} = \lfloor s_{t} \rfloor > \lfloor s_{t-1} \rfloor,$ in which case $\lfloor s_{t} \rfloor = \lfloor s_{t-1} \rfloor + 1$. Again, non-negativity of $p_t$ is immediate, while $$p_t = \frac{s_t - \lfloor s_t \rfloor}{s_{t-1} - \lfloor s_{t-1} \rfloor} = \frac{s_{t-1}+x_t - \lfloor s_{t-1} \rfloor - 1}{s_{t-1} - \lfloor s_{t-1} \rfloor} \leq \frac{s_{t-1} - \lfloor s_{t-1} \rfloor}{s_{t-1} - \lfloor s_{t-1} \rfloor} = 1.$$

	To prove Property \ref{prop:marginals}, we will need a closed form for $q_t := \Pr[S_{t} = \lfloor s_t \rfloor]$, useful when discussing $\Pr[S_{t-1} = \lfloor s_{t} \rfloor]$, the probability that $t$ may be added to $\calS$, and $\Pr[S_{t-1} < \lfloor s_t \rfloor]$, the probability that $t$ \emph{must} be added to $\calS$ (to satisfy Property \ref{prop:floor/ceil}).
	Note that $S_t \sim \lfloor s_t \rfloor + \textrm{Ber}(1-q_t)$ by Property \ref{prop:floor/ceil}, while $\E[S_t] = \sum_{t'\leq t} \E[X_{t'}] = \sum_{t'\leq t} x_{t'} = s_t$, by Property \ref{prop:marginals} (for all $t'\leq t$) and linearity of expectation. Combining these observations, we obtain $s_t = \E[S_t] = \lfloor s_t \rfloor + 1-q_t$, 
	and therefore, 
	\begin{align}\label{eqn:q_t}
	q_t = \lfloor s_t \rfloor + 1 - s_t.
	\end{align}
	We now use the closed form for $q_{t-1}$ to prove Property \ref{prop:marginals} for time $t$. 
	The easy case is when $\lfloor s_t \rfloor = \lfloor s_{t-1} \rfloor$. Here we have that $\E[X_t \mid S_{t-1} = \lfloor s_{t} \rfloor] = \frac{x_t}{q_{t-1}}$, while trivially $\E[X_t \mid S_{t-1} = \lceil s_{t} \rceil] = 0$. Consequently, since $S_{t-1}\in \{\lfloor s_{t-1}\rfloor, \lceil s_{t-1} \rceil\} = \{\lfloor s_{t}\rfloor, \lceil s_{t} \rceil\}$ in this case, then by total probability we have that $\E[X_t] = q_{t-1}\cdot \frac{x_t}{q_{t-1}} + 0 = x_t$, as desired.

	Finally, we consider the case where $\lfloor s_{t} \rfloor > \lfloor s_{t-1} \rfloor$, i.e., $\lfloor s_{t} \rfloor = \lfloor s_{t-1} \rfloor + 1$. 
	If in addition $s_{t-1}=\lfloor s_{t-1} \rfloor$ (i.e., $s_{t-1}$ is integral, implying that $s_t = s_{t-1}+1$, and hence $x_t = s_t - s_{t-1} = 1$), we have by Property \ref{prop:floor/ceil} that $S_{t-1}=s_{t-1} < \lfloor s_t \rfloor$ (with probability one), and so we add $t$ to $\calS$ with probability one, resulting in $\E[X_t] = 1 = x_t$, as desired.	
	If, conversely, we  have that $s_{t-1}\neq \lfloor s_{t-1}\rfloor < \lfloor s_{t}\rfloor$, then by total probability, we obtain the desired equality as follows.
	\begin{align*}
	\E[X_t] & = \Pr[S_{t-1} = \lfloor s_{t-1}\rfloor]\cdot 1 + \Pr[S_{t-1} \neq \lfloor s_{t-1}\rfloor] \cdot \frac{s_t - \lfloor s_t \rfloor}{s_{t-1} - \lfloor s_{t-1} \rfloor} \\
	& = q_{t-1} + (1-q_{t-1})\cdot \frac{x_t + s_{t-1} - (\lfloor s_{t-1} \rfloor + 1)}{s_{t-1} - \lfloor s_{t-1} \rfloor} && \begin{cases} s_t=x_t + s_{t-1} \\ \lfloor s_t\rfloor = \lfloor s_{t-1}\rfloor +1
	\end{cases} \\
	& = q_{t-1} + (1-q_{t-1})\cdot \frac{x_t - q_{t-1}}{1-q_{t-1}} && \textrm{\Cref{eqn:q_t}}\\
	& = x_t. && \qedhere
	\end{align*}
\end{proof}

\subsection{Coupling the offline and online algorithms}

Here we prove that our online level-set algorithm, \Cref{alg:rounding}, induces the same output distribution as the offline algorithm of \cite{srinivasan2001distributions}, \Cref{alg:linear-srinivasan-rounding}
\equivalence*
\begin{proof}
We prove by induction on $t\in [n]$ that the following holds:
\begin{equation}
    \label{eqn:online-equals-offline}
    \forall (b_1, \ldots, b_{t-1}) \in \{0,1\}^{t-1}, ~
\Pr_\calD[X_t = 1 \bigm| (\forall i < t, ~X_i = b_i)] =
\Pr_{\calD'}[X_t = 1 \bigm| (\forall i < t, ~X_i = b_i)]. 
\end{equation}
The (strong) inductive hypothesis, whereby \Cref{eqn:online-equals-offline} holds for all $t'< t$ clearly implies that for all $(b_1,\dots,b_t)\in \{0,1\}^{t-1}$ we have that $\Pr_\calD[(\forall i<t, ~X_i=b_i)]=\Pr_{\calD'}[(\forall i<t, ~X_i=b_i)]$.
Similarly, \Cref{eqn:online-equals-offline} holds for all $t\leq n$ iff the two distributions are the same.

The base case $t = 1$ follows from Property \ref{prop:marginals}. So suppose $t > 1$. Fix $b_1, \ldots, b_{t-1}\in \{0,1\}$, and let $\mathcal{E}$ denote the event ``$\forall i < t, ~X_i = b_i$" (measured w.r.t.\ 
$\calD$ or w.r.t.\ $\calD'$). Then, by the strong induction hypothesis, 
$\Pr_\calD[\mathcal{E}] = \Pr_{\calD'}[\mathcal{E}]$, and in particular,
$\Pr_\calD[\mathcal{E}] \not= 0$ iff $\Pr_{\calD'}[\mathcal{E}] \not= 0$; thus we may assume that the events conditioned on in the LHS and in the RHS of \eqref{eqn:online-equals-offline} both have nonzero (identical) probabilities. This is the only place where we will need the induction hypothesis. 

Recall that $s_{t-1} = \sum_{i=1}^{t-1} x_i$. Furthermore, although $S_{t-1}$ technically depends on the rounding algorithm, it is natural to take it to be $\sum_{i=1}^{t-1} b_i$ here. Note that since both the offline and online algorithms satisfy Property \ref{prop:floor/ceil}, we have that
$S_{t-1} \in \{\lfloor s_{t-1} \rfloor,\lceil s_{t-1} \rceil\}$. 

\smallskip \noindent \textbf{Case I: $s_{t-1} \in \mathbb{Z}$.} This is an easy case. It is immediate that 
$\Pr_\calD[X_t = 1 \bigm| \mathcal{E}] = x_t$ here. This is also not hard to see this for the offline algorithm. 
Recall that the offline algorithm starts with a copy $y$ of $x$. 
Since $s_{t-1} \in \mathbb{Z}$, the offline algorithm would set $X_i$ for exactly $s_{t-1}$ indices $i \in [t-1]$ to one (and the rest in $[t-1]$ to zero), and would then recursively run on the suffix $(y_t, y_{t+1}, \ldots, y_n)$ of $y$ with no correlation with the rounding thus far. Furthermore, $y_j = x_j$ for all $j \geq t$. Thus, by Property \ref{prop:marginals} applied to this suffix, we have that 
$\Pr_{\calD'}[X_t = 1 \bigm| \mathcal{E}] = x_t$ also.

\smallskip 
We may thus assume from now on that $z \doteq s_{t-1} - \lfloor s_{t-1} \rfloor$ lies in $(0,1)$.

\smallskip \noindent \textbf{Case II: $s_{t-1} \not\in \mathbb{Z}$ and $z + x_t \leq 1$.} An easy sub-case here is that 
$S_{t-1} = \lceil s_{t-1} \rceil$; here, since both the offline and online algorithms satisfy Property \ref{prop:floor/ceil}, it is immediate that in this sub-case
\[ \Pr_\calD[X_t = 1 \bigm| \mathcal{E}] = \Pr_{\calD'}[X_t = 1 \bigm| \mathcal{E}] = 0. \]
We may thus assume that 
$S_{t-1} = \lfloor s_{t-1} \rfloor$. By the definition of the online algorithm \ref{alg:rounding} we have that
\begin{equation}
    \label{eqn:caseii-online}
 \Pr_\calD[X_t = 1 \bigm| \mathcal{E}] = \frac{x_t}{1 - z}.   
\end{equation}
We now analyze $\Pr_{\calD'}[X_t = 1 \bigm| \mathcal{E}]$. Consider the offline algorithm \emph{just before} it involves $x_t$ in STEP: at this point in time, it has:
\begin{itemize}
\item set $\lfloor s_{t-1} \rfloor$ of the variables $(X_i: i < t)$ to one;
\item set $t - 2 - \lfloor s_{t-1} \rfloor$ of the variables $(X_i: i < t)$ to zero; and, recalling again that the offline algorithm starts with a copy $y$ of $x$,
\item it has assigned the current value $z$ to one variable $y_{i^*}$ where $i^* \in [t-1]$ is a random variable from some distribution. 
\end{itemize}
At this point, the offline algorithm is basically recursively run on the sequence $(y_{i^*}, y_t, y_{t+1}, \ldots, y_n)$, where $y_{i^*} = z$ and $y_j = x_j$ for all $j \geq t$. Crucially, since $S_{t-1} = \lfloor s_{t-1} \rfloor$, the conditioning on $\mathcal{E}$ says that the variable $y_{i^*}$ is \emph{eventually rounded to $0$}. Thus, we can compute $\Pr_{\calD'}[X_t = 1 \bigm| \mathcal{E}]$ as follows. Suppose the offline algorithm is run on the sequence $(u_1, u_2, \ldots, u_{n-t+2}) \doteq (y_{i^*}, y_t, y_{t+1}, \ldots, y_n)$, to output a random bit-vector $(U_1, U_2, \ldots, U_{n-t+2})$; then,
\begin{equation}
    \label{eqn:caseii-offline}
  \Pr_{\calD'}[X_t = 1 \bigm| \mathcal{E}] = \Pr_{\calD'}[U_2 = 1 \bigm| U_1 = 0]. 
\end{equation}
But since $u_1 + u_2 = z + x_t \leq 1$, the first STEP applied to $u_1$ and $u_2$ ensures that \emph{at most one of $u_1$ and $u_2$ gets rounded to one}. This yields
\[ \Pr_{\calD'}[U_2 = 1 \bigm| U_1 = 0] = 
\frac{Pr_{\calD'}[(U_2 = 1) \wedge (U_1 = 0)]}{Pr_{\calD'}[U_1 = 0]} 
= \frac{Pr_{\calD'}[U_2 = 1]}{Pr_{\calD'}[U_1 = 0]}  = \frac{x_t}{1 - z}, \]
where the second equality holds since at most one of $U_1$ and $U_2$ is one; we are thus done by \eqref{eqn:caseii-online}. 

\smallskip \noindent \textbf{Case III: $s_{t-1} \not\in \mathbb{Z}$ and $z + x_t > 1$.} This is handled similarly. The easy sub-case here is that 
$S_{t-1} = \lfloor s_{t-1} \rfloor$; here, since both the offline and online algorithms satisfy Property \ref{prop:floor/ceil}, it is immediate that
$$\Pr_\calD[X_t = 1 \bigm| \mathcal{E}] = \Pr_{\calD'}[X_t = 1 \bigm| \mathcal{E}] = 1.$$
We may thus assume that 
$S_{t-1} = \lceil s_{t-1} \rceil$. By the definition of the online algorithm \ref{alg:rounding} we have that
\begin{equation}
    \label{eqn:caseiii-online}
 \Pr_\calD[X_t = 1 \bigm| \mathcal{E}] = \frac{x_t+z-1}{z}.   
\end{equation}
Now, the offline algorithm is, like in Case II, essentially run on the sequence $(u_1, u_2, \ldots, u_{n-t+2}) \doteq (y_{i^*}, y_t, y_{t+1}, \ldots, y_n)$---where $y_{i^*}$ is initially $z$---to output a random bit-vector $(U_1, U_2, \ldots, U_{n-t+2})$, \emph{with the key difference from Case II being that $U_1$ will eventually be set to $1$}. So, 
\begin{equation}
    \label{eqn:caseiii-offline}
  \Pr_{\calD'}[X_t = 1 \bigm| \mathcal{E}] = \Pr_{\calD'}[U_2 = 1 \bigm| U_1 = 1]. 
    \end{equation}
Now, since $u_1 + u_2 = z + x_t > 1$, the first STEP applied to $u_1$ and $u_2$ ensures that \emph{at least} one of $u_1$ and $u_2$ gets rounded to one. This implies 
\begin{eqnarray*}
\Pr_{\calD'}[U_2 = 1 \bigm| U_1 = 1] & = & 
\frac{Pr_{\calD'}[U_2 = 1] - Pr[(U_2 = 1) \wedge (U_1 = 0)]}{Pr_{\calD'}[U_1 = 1]} \\
& = & \frac{Pr_{\calD'}[U_2 = 1] - Pr_{\calD'}[U_1 = 0]}{Pr_{\calD'}[U_1 = 1]}
\\
& = & \frac{x_t + z - 1}{z},
\end{eqnarray*}
where the second equality holds since at least one of $U_1$ and $U_2$ is one;
we are thus done by \eqref{eqn:caseiii-online}. This concludes the proof of the inductive step.
\end{proof}

	\section{Deferred Proofs of \Cref{sec:lower-bound}}\label{appendix:lower-bounds}

In this section we prove that any sufficiently many binary variables must contain some $2^r$ with some (near-)positive $2^r$-wise correlation, generalizing the special case of $r=1$ given by \Cref{pairwise-positive-correlation}.
\almostpositive*
\begin{proof}
	We prove this claim for all $0\leq \epsilon\leq p \leq 1$ by induction on $r\geq 1$, and note that the claim is trivial if $p^{2^r}\leq \eps$, in which case the RHS is negative, so $n_r(p,\eps)=2^r$ suffices.
        \Cref{pairwise-positive-correlation} proves the base case, showing that $n_1(p,\epsilon) \leq \lceil \frac{2p}{\epsilon}+1\rceil$ suffices. To prove the inductive step, 
        consider a set of $k:=n_1(p,\frac{\epsilon}{2^{2^{r}}})+2m$ variables $Y_1,\dots, Y_k\sim \Ber(p)$, for $m$ to be determined shortly. Then, we find a sequence of disjoint pairs $I_1,I_2,\dots,I_m\subseteq [n]$, where for each $I_s=\{i,j\}$, $s\in [m]$, we have that $\Cov(Y_i,Y_j)\geq -\frac{\epsilon}{2^{2^{r}}}$.
        The existence of such pairs follows by repeatedly invoking  \Cref{pairwise-positive-correlation} applied to the variables $\{Y_i \mid i\in [n]\setminus \bigcup_{\ell=1}^s {I_\ell}\}$ to obtain the pair $I_{s+1}$. 
	Now, for each pair $I_s = \{i,j\}$ as above, define a new variable $Z_s := Y_{i}\cdot Y_j$. By $\Cov(Y_i,Y_j)\geq -\frac{\epsilon}{2^{2^{r}}}$, we have that $\Pr[Z_s] \geq p^2 - \frac{\epsilon}{2^{2^{r}}}$.
	Next, couple these combined variables with equiprobable variables $A_s \sim \Ber(p^2 -\frac{\epsilon}{2^{2^{r}}})$ with $Z_s \geq A_s$ always.
        Then, if we take $m = n_{r-1}(p^2-\frac{\epsilon}{2^{2^{r}}},\frac{\epsilon}{2})$, the inductive hypothesis applied to the $m$ variables $A_s$ implies the existence of a subset $S\subseteq [m]$ of $2^{r-1}$ of these variables satisfing the following.
	\begin{align*}
		\E\left[\prod_{i\in \bigcup_{s\in S} I_s} Y_i \right] & = \E\left[\prod_{s\in S} Z_s \right] \geq 
		\E\left[\prod_{s\in S} A_s \right] \geq \left(p^2-\frac{\epsilon}{2^{2^{r}}}\right)^{2^{r-1}} - \frac{\epsilon}{2} \geq p^{2^r} - 2^{2^{r-1}}\cdot \frac{\epsilon}{2^{2^{r}}} - \frac{\epsilon}{2} \geq p^{2^r} - \epsilon.
	\end{align*}
        Above, the first inequality follows by the coupling $Z_s\geq A_s$. The second inequality follows from the inductive hypothesis. The third inequality follows from $(a-b)^k \geq a^k - 2^k\cdot b$ for any $a,b\in [0,1]$ and $k\in \mathbb{N}$ (as seen by expanding $(a-b)^k$). The final inequality follows from $2^{2^{r}-1}\geq 2^{2^{r-1}}$ for $r\geq 1$.
	We conclude that $\{Y_i \mid i\in \bigcup_{s\in S} I_s\}$ is the desired subset of $2^r$ variables in $Y_1,\dots,Y_n$.
\end{proof}

\begin{remark}
	We did not attempt to optimize the minimum possible $n_r(p,\eps)$ above. Nonetheless, for concrete bounds, we note that if $p^{2^r}\geq \eps$, then our proof yields bounds satisfying the following recurrence.
	\begin{align*}
		n_1(p,\epsilon) & \leq \left\lceil\frac{2p}{\epsilon} + 1 \right\rceil \leq \frac{4p}{\epsilon}, \\
		n_r(p,\epsilon) & \leq n_1\left(p,\frac{\epsilon}{2^{2^{r}}}\right) + 2n_{r-1}\left(p^2 - \frac{\epsilon}{2^{2^{r}}},\frac{\epsilon}{2}\right) \leq 
		\frac{4\cdot 2^{2^{r}} p}{\epsilon} + 2n_{r-1}\left(p^2 - \frac{\epsilon}{2^{2^{r}}},\frac{\epsilon}{2}\right).
	\end{align*}
	Noting that this recurrence is monotone increasing in $p$, we have that $n_r(p,\eps)=O\left(2^{2^r}\cdot \frac{p}{\eps}\right)$.
\end{remark}

	\bibliographystyle{acmsmall}
	\bibliography{abb,ultimate}
	\end{document}